\documentclass[runningheads,a4page]{llncs}
\usepackage{amssymb}
\usepackage{amsmath}
\usepackage[usenames]{color}
\usepackage{pgf}
\usepackage{paralist}
\usepackage{tikz}
\usetikzlibrary{arrows,automata}
\usepackage{times}
\usepackage{cases}

\def\phi{\varphi}

\DeclareMathOperator{\U}{\text{\sf U}}
\DeclareMathOperator{\X}{\text{\sf X}}

\newcommand{\TRAN}[1]{\xrightarrow[]{#1}}
\newcommand{\TRANP}[1]{\xrightarrow[]{#1}_{\text{P}}}

\newcommand{\AP}{\mathit{AP}}

\newcommand{\WBS}{\approx}

\newcommand{\iBS}[1]{\sim_{#1}}

\newcommand{\DIRAC}[1]{\mathcal{D}_{#1}}
\newcommand{\TR}{\mathit{tr}}

\newcommand{\RPLUS}{\mathit{R}^{>0}}
\newcommand{\Real}{\mathit{R}^{\ge 0}}
\newcommand{\DIST}{\mathit{Dist}}
\newcommand{\MC}[1]{\mathcal{#1}}
\newcommand{\XI}[1][I]{\X^{#1}}
\newcommand{\UI}[1][I]{\U^{#1}}

\newcommand{\Path}[1]{\mathit{Paths}^{#1}}
\newcommand{\Len}[1]{|#1|}
\newcommand{\iPath}[2]{#1[#2]}
\newcommand{\lsPath}[1]{#1\!\downarrow}
\newcommand{\Time}{\mathit{time}}
\newcommand{\CONCAT}{\!\!\smallfrown\!\!}
\newcommand{\AT}{@}
\newcommand{\NEXT}{\mathit{Steps}}
\newcommand{\RATE}{\lambda}
\newcommand{\Sch}{\pi}
\newcommand{\Field}{\mathfrak{F}}
\newcommand{\PM}{\mathit{Pr}}
\newcommand{\PRE}[2][\omega]{#1|^{#2}}
\newcommand{\SUF}[2][\omega]{#1|_{#2}}

\newcommand{\INFPATH}{\infty}
\newcommand{\SUPP}{\mathit{Supp}}
\newcommand{\CSL}{\text{CSL}}

\newcommand{\CTMDP}{\text{CTMDP}}
\newcommand{\CTMC}{\text{CTMC}}
\newcommand{\MA}{\text{MA}}
\newcommand{\PA}{\text{PA}}
\newcommand{\MDP}{\text{MDP}}
\newcommand{\IMC}{\text{IMC}}
\newcommand{\PCTL}{\text{PCTL}}
\newcommand{\CTL}{\text{CTL}}
\newcommand{\MTL}{\text{MTL}}
\newcommand{\CSLstar}{\text{CSL}^{*}}
\newcommand{\INTER}[2]{#1\parallel#2}
\newcommand{\PAR}[2]{#1\parallel#2}
\newcommand{\SAT}{\mathit{Sat}}

\newcommand{\MI}[1]{\mathit{#1}}
\newcommand{\ESET}[2][\MC{R}]{[#2]_{#1}}

\newcommand{\ABSORB}[1]{#1_{\!\bot}}

\newcommand{\UNIFORM}[1]{\bar{#1}}
\newcommand{\ABS}[1]{|#1|\ }
\newcommand{\SUC}{\MI{Suc}}

\begin{document}
\title{Bisimulations and Logical Characterizations on Continuous-time Markov Decision Processes}
\titlerunning{Bisimulations and Logical Characterizations on $\CTMDP$s}
\author{Lei Song\inst{1}
\and
Lijun Zhang\inst{2} \and Jens Chr. Godskesen\inst{3}
}

\institute
{
\inst{}
  Max-Planck-Institut f\"{u}r Informatik and   Saarland University, Saarbr\"{u}cken, Germany
  \and
  \inst{}
  State Key Laboratory of Computer Science, 
Institute of Software, Chinese Academy of Sciences 
  \and
 \inst{}%
  Programming, Logic, and Semantics Group,
  IT University of Copenhagen, Denmark  
}

\maketitle

\begin{abstract}
  In this paper we study strong and weak bisimulation equivalences for
  continuous-time Markov decision processes ($\CTMDP$s) and the logical
  characterizations of these relations with respect to the
  continuous-time stochastic logic ($\CSL$). For strong bisimulation,
  it is well known that it is strictly finer than $\CSL$
  equivalence. In this paper we propose strong and weak bisimulations for
  $\CTMDP$s and show that for a subclass of $\CTMDP$s, strong and weak
  bisimulations are both sound and complete with respect to the
  equivalences induced by $\CSL$ and the sub-logic  of $\CSL$ without next
  operator respectively. We then consider a standard extension of $\CSL$,
  and show that it and its sub-logic without $\X$ 
  can be fully characterized by strong and weak bisimulations 
  respectively over arbitrary $\CTMDP$s.
\end{abstract}

\section{Introduction}
Recently, continuous-time Markov decision processes ($\CTMDP$s) have
received extensive attention in the model checking community, see for
example
\cite{BaierHKH05,WolovickJ2007,MartinK2007,NeuhausserSK09,BuSc10,RabeS11}.
%,BrazdilFKKK09,NeuhausserZ10,BuchholzHHZ11,fssttcs}.
Analysis techniques for $\CTMDP$s suffer especially from the state space explosion problem. 
Thus, as for other stochastic models,
bisimulation relations have been proposed for $\CTMDP$s. In
\cite{MartinK2007}, strong bisimulation was shown to be sound with respect to the
continuous-time stochastic logic~\cite{Aziz1996VCT} ($\CSL$). This result guarantees that
one can first reduce a $\CTMDP$ up to bisimulation equivalence before analysing it.
On the other hand,
as indicated in~\cite{MartinK2007}, strong bisimulation is
not complete with respect to $\CSL$, i.e., logically equivalent states
might be not bisimilar.

$\CTMDP$s extend Markov decision processes ($\MDP$s) with
exponential sojourn time distributions, and subsume models such as
labelled transition systems and Markov chains as well.  While linear
and branching time equivalences have been studied for
these sub-models \cite{Glabbeek1,Glabbeek93,BaierKHW05,SongZG2011},
we extend these results to the setting of $\CTMDP$s.  In this paper we study
strong and weak bisimulation relations for $\CTMDP$s, and the
logical characterization problem of these relations with respect to $\CSL$ and its sub-logics.

We start with a slightly coarser notion of strong bisimulation than
the one in~\cite{MartinK2007}, and then propose \emph{weak
  bisimulation} for $\CTMDP$s. We study the relationship between
strong and weak bisimulations and the logical equivalences induced by
$\CSL$ and $\CSL_{\backslash \!\X}$ -- the sub-logic of $\CSL$ without
next operators. Our first contribution is to identify a subclass of
$\CTMDP$s under which our strong and weak bisimulations coincide with
$\CSL$ and $\CSL_{\backslash \!\X}$ equivalences respectively. We
discuss then how this class of $\CTMDP$s can be efficiently
determined, and moreover, we argue that most models arising in 
practice are among this class.

As for labelled transition systems and $\MDP$s, we also 
define an extension of $\CSL$, called $\CSLstar$, which 
is more distinguishable than $\CSL$. Surprisingly, $\CSLstar$ is able to
fully characterize strong bisimulation over arbitrary $\CTMDP$s, similarly for the
sub-logic without next operator and weak bisimulation.

Since $\CTMDP$s can be seen as models combining $\MDP$s and continuous-time Markov Chains
($\CTMC$s), we will discuss the downward compatibility of the relations
with those for $\MDP$s~\cite{SegalaL95} and $\CTMC$s in
\cite{BaierKHW05}.
Summarizing, the paper contains the following contributions:
\begin{enumerate}
\item We extend strong probabilistic bisimulation defined in~\cite{SegalaL95}
over probabilistic automata to $\CTMDP$s, and then prove that
it coincides with $\CSL$ equivalence for a subclass of $\CTMDP$s;
\item We propose a scheme to determine the subclass of $\CTMDP$s efficiently, and
show that many models in practice are in this subclass;
\item We introduce a new notion of weak bisimulation for $\CTMDP$s, 
and show its characterization results with respect to $\CSL_{\backslash\!\X}$;
\item We present a standard extension of $\CSL$ that is shown to be both sound and complete
 with respect to strong and weak bisimulations for arbitrary $\CTMDP$s.
\end{enumerate}

\paragraph{Related work.}
Logical characterizations of bisimulations have been studied
extensively for stochastic models.  For $\CTMC$s, $\CSL$ characterizes
strong bisimulation, while $\CSL$ without next operator characterizes
weak bisimulation~\cite{BaierKHW05}. Our results in this paper are
conservative extensions for both strong and weak bisimulations from $\CTMC$s
to $\CTMDP$s.
In~\cite{DesharnaisP03}, the results are extended to $\CTMC$s with
continuous state spaces.

For $\CTMDP$s, the first logical characterization result is presented
in \cite{MartinK2007}. It is shown that strong bisimulation is sound, but not
complete with respect to $\CSL$ equivalence. 
In this paper, we introduce strong and weak bisimulation relations for
$\CTMDP$s. For a subclass of $\CTMDP$s, i.e., those without \textit{2-step recurrent}
states, we show that they are also complete for $\CSL$ and
$\CSL_{\backslash\!\X}$ equivalences respectively.

For probabilistic automata ($\PA$s), Hennessy-Milner logic has been
extended to characterize bisimulations
in~\cite{Jonsson,DArgenioWTC09,HermannsPSWZ11}.
In~\cite{DesharnaisGJP10}, Desharnais \emph{et al.} have shown that
weak bisimulation agrees with $\PCTL^*$ equivalence for alternative $\PA$s.  
Another related paper for $\PA$s is our previous paper \cite{SongZG2011}, in
which we have introduced $i$-depth bisimulations to
characterize logical equivalences induced by $\PCTL^*$ and its
sub-logics. 

All proofs are found in the full version of this paper~\cite{SongBTS2012}.

\paragraph*{Organization of the paper.}
Section~\ref{sec:pre} recalls the definition of $\CTMDP$s and the
logic $\CSL$. Variants of bisimulation relations and their
corresponding logical characterization results are studied in Section
\ref{sec:bisimulation}. In Section~\ref{sec:mtl} we present
the extension of $\CSL$ that fully characterizes strong and weak
bisimulations. We discuss in Section~\ref{sec:relation} related work
with $\MDP$s and $\CTMC$s.  Section~\ref{sec:conclusion} concludes the
paper.

\section{Preliminaries}\label{sec:pre}
%\paragraph{Distributions.}
For a finite set $S$, a distribution is  a function $\mu:S\to [0,1]$ satisfying $\ABS{\mu}:=\sum_{s\in S}\mu(s)= 1$.
We denote by $\DIST(S)$ the set of distributions
over $S$.  We shall use $s,r,t,\ldots$ and $\mu,\nu\ldots$ to range
over $S$ and $\DIST(S)$, respectively. The support of $\mu$ is
defined by $\SUPP(\mu)=\{s\in S \mid \mu(s)>0\}$. 
Given a finite set of non-negative real numbers $\{p_j\}_{j\in J}$ and distributions $\{\mu_j\}_{j\in J}$ 
such that $\sum_{j\in J}p_i=1$ for each $j\in J$, 
$\sum_{j\in J}p_j\cdot\mu_j$ is the
distribution such that $(\sum_{j\in J}p_j\cdot\mu_j)(s)=\sum_{j\in J}p_j\cdot\mu_j(s)$ for each
$s\in S$.
For an equivalence relation $\mathcal{R}$ over $S$, 
we write $\mu~\mathcal{R}~\nu$ if it holds that $\mu(C)=\nu(C)$ for all equivalence classes 
$C\in S/\mathcal{R}$ where $\mu(C)=\sum_{s\in C}\mu(s)$,
and moreover $\ESET{s}=\{r\mid s~\MC{R}~r\}$ is the equivalence class of 
$S/\MC{R}$ containing $s$. The subscript $\MC{R}$ will be omitted if it is clear from the context. A 
distribution $\mu$ is called \emph{Dirac} if $|\SUPP(\mu)|=1$, and we let $\DIRAC{s}$ denote the Dirac 
distribution such that $\DIRAC{s}(s)=1$. 
We let $\Real$ and $\RPLUS$ denote the set of non-negative and positive real numbers respectively.
%Let $\MC{R}$ be a relation over $S$, define
%$\UPWARD{\MC{R}}{C}=\{r\mid s~\MC{R}~r\land s\in C\}$ and $\DOWNWARD{\MC{R}}{C}=\{r\mid %r~\MC{R}~s\land s\in C\}$. We say $C$ is $\MC{R}$ \emph{upward closed} iff $C=\UPWARD{\MC{R}}{C}$, %and similarly $C$ is $\MC{R}$ \emph{downward closed} iff $C=\DOWNWARD{\MC{R}}{C}$. 

\subsection{Continuous-time Markov Decision Processes}
Below follows the definition of $\CTMDP$s, which subsume both $\MDP$s and $\CTMC$s.  

\begin{definition}[Continuous-time Markov Decision Processes]\label{def:ctmdp}
A tuple $\MC{C}=(S, \rightarrow, \AP, L, s_0)$ is a \emph{$\CTMDP$} where $s_0\in S$ is the initial state,
$S$ is a finite but non-empty set of states,
$\AP$ is a finite set of atomic propositions, 
$L:S\mapsto 2^{\AP}$ is a labelling function, and
$\rightarrow\subseteq S\times\RPLUS\times\DIST(S)$ is a finite transition relation 
such that for each $s\in S$, there exists $\lambda$ and $\mu$ with $(s,\lambda,\mu)\in\rightarrow$.
\end{definition}
From Definition~\ref{def:ctmdp} we can see that
there are both non-deterministic and probabilistic transitions in a $\CTMDP$.
We write $s\TRAN{\lambda}\mu$ if $(s,\lambda,\mu)\in\ \rightarrow$, where 
$\lambda$ is called exit rate of the transition.
Let $\SUC(s)=\{r\mid\exists(s\TRAN{\lambda}\mu).\mu(r)>0\}$ denote the
successor states of $s$, and let $\SUC^*(s)$ be its transitive closure. 
A state $s$ is said to be \emph{silent} iff for all
$s_1,s_2\in\SUC^*(s)$, $L(s_1)=L(s_2)$ and $s_1\TRAN{\lambda}\mu_1$ implies $s_2\TRAN{\lambda}\mu_2$.
Intuitively, a state $s$ is silent if all its reachable states have the same labels as $s$. In addition,
they have transitions with the same exit rates as transitions of $s$. States like $s$ are called silent, since
it is not distinguishable from all its successors, either by labels or sojourn time
of states. Therefore a silent state $s$ and all its successors can be represented 
by a single state which is the same as $s$ but with all its outgoing transitions 
leading to itself.
A $\CTMC$ is a deterministic $\CTMDP$ satisfying the
condition: $s\TRAN{\lambda}\mu$ and $s\TRAN{\lambda'}\mu'$ imply
$\lambda=\lambda'$ and $\mu=\mu'$ for any $s\in S$.

\subsection{Paths, Uniformization, and Measurable Schedulers}
Let $\MC{C}=(S,\rightarrow,\AP,L,s_0)$ be a $\CTMDP$ fixed for the remainder of the paper.
Let $\Path{n}(\MC{C})=S\times(\RPLUS\times S)^n$ denote the set containing paths of 
$\MC{C}$ with length $n$. The set of all finite paths of
$\MC{C}$ is the union of all finite paths $\Path{*}(\MC{C})=\cup_{n\geq0}\Path{n}(\MC{C})$. Moreover,
$\Path{\INFPATH}(\MC{C})=S\times(\RPLUS\times S)^{\INFPATH}$ contains
all infinite paths and
$\Path{}(\MC{C})=\Path{*}(\MC{C})\cup\Path{\INFPATH}(\MC{C})$ is the
set of all (finite and infinite) paths of $\MC{C}$. 
Intuitively, a path is comprised of an alternation of states and their sojourn time. 
To simplify the discussion we introduce some  notations. Given a path
$\omega=s_0,t_0,s_1,t_1\cdots s_{n}\in\Path{n}(\MC{C})$,
$\Len{\omega}=n$ is the length of $\omega$, $\lsPath{\omega}=s_{n}$ is
the last state of $\omega$, $\PRE{i}=s_0,t_0,\cdots,s_i$ is the prefix
of $\omega$ ending at the $(i+1)$-th state, and
$\SUF{i}=s_i,t_i,s_{i+1},\cdots$ is the suffix of $\omega$ starting
from the $(i+1)$-th state, and  $\omega\CONCAT(t_{n},s_{n+1})$ is the
path obtained by extending $\omega$ with $(t_{n},s_{n+1})$. Let $\iPath{\omega}{i}=s_i$
denote the $(i+1)$-th state where $i\le n$ and $\Time(\omega,i)=t_i$ the sojourn time in the
$(i+1)$-th state with $i<n$. Let $\omega\AT t$ be the state
at time $t$ in $\omega$, that is, $\omega\AT
t=\iPath{\omega}{j}$ where $j$ is the smallest index such that
$\sum_{i=0}^j t_i>t$. Moreover,
$\NEXT(s)=\{(\RATE,\mu)\mid(s,\RATE,\mu)\in\ \rightarrow\}$ is
the set of all available choices at state $s$. 
Let $\{I_i\subseteq[0,\infty)\}_{0\leq i\leq k}$ denote a set of non-empty closed intervals, 
then $C(s_0,I_0,\cdots,I_{k},s_{k+1})$ is the \emph{cylinder set} of paths
$\omega\in\Path{\INFPATH}(\MC{C})$ such that $\iPath{\omega}{i}=s_i$ for 
$0\le i\le k+1$
and $\Time(\omega,i)\in I_i$ for $0\le i\le k$. Let $\Field_{\Path{\INFPATH}(\MC{C})}$
be the smallest $\sigma$ algebra on $\Path{\INFPATH}(\MC{C})$
containing all cylinder sets.

As shown in~\cite{Baier2003MAC}, model checking of $\CTMC$s 
can be reduced to the problem of computing transient state probabilities,
which can be solved efficiently, for instance by uniformization.
In a uniformized $\CTMC$, all states will evolve at the same speed, i.e.,
all transitions have the same exit rates.   
Similarly, we can also define uniformization of a $\CTMDP$ by uniformizing
the exit rate of all its transitions. 
Below we recall the notion of \emph{uniformization} for $\CTMDP$s \cite{BuSc10,NeuhausserSK09}.
\begin{definition}[Uniformization]
\label{def:uniformization}
Given a \emph{$\CTMDP$} $\MC{C}=(S,\rightarrow,\AP,L,s_0)$, the uniformized \emph{$\CTMDP$} is 
denoted as $\UNIFORM{\MC{C}}=(\UNIFORM{S},\rightarrow',\AP,\UNIFORM{L},\UNIFORM{s_0})$ 
where 
\begin{enumerate}
\item $\UNIFORM{S}=\{\UNIFORM{s}\mid s\in S\}$,  $\UNIFORM{s_0}\in\UNIFORM{S}$ is the initial state,
\item $\UNIFORM{L}(\UNIFORM{s})=L(s)$ for each $s\in S$, and
\item $(\UNIFORM{s},E,\UNIFORM{\mu})\in\rightarrow'$ iff there exists 
$(s,\lambda,\mu)\in\rightarrow$ and $\UNIFORM{\mu}=\frac{\lambda}{E}\cdot\mu' + (1-\frac{\lambda}{E})\cdot\DIRAC{\UNIFORM{s}}$ such that $\mu'(\UNIFORM{r})=\mu(r)$ for each $r\in\SUPP(\mu)$,
\end{enumerate}
Here  $E$ 
is the uniformization rate for $\UNIFORM{\MC{C}}$, which is 
a real number equal or greater than all the rates appearing in $\MC{C}$. 
%A $\CTMDP$ $\MC{C}$ is uniformized iff for any $(s_1,\lambda_1,\mu_1)\in\rightarrow$ and $(s_2,\lambda_2,\mu_2)\in\rightarrow$, $\lambda_1=\lambda_2$.
\end{definition}
By uniformization for each transition $(s,\lambda,\mu)$ we add a self loop to $s$ with rate equal to 
$E$ minus the original rate $\lambda$.
After uniformization every state will have a unique exit rate on all its transitions. 
As we will show later, this transformation
will not change the properties we are interested in 
under certain classes of schedulers.

Due to the existence of non-deterministic choices in $\CTMDP$s, we need to resolve
them to define
probability measures. As usual, 
non-deterministic choices in $\CTMDP$s are resolved by schedulers (or policies or adversaries), 
which generate a distribution over the available transitions based on the
given history information. Different classes of schedulers can be defined depending on
the information a scheduler can use in order to choose the next transition. 
However not all of them are suitable for our purposes, which we will explain later.
In this paper, we shall focus on one specific class of schedulers, 
called \textit{measurable total time positional schedulers} (TTP)~\cite{NeuhausserSK09},
which is defined as follows:
\begin{definition}[Schedulers]\label{def:scheduler}
A scheduler $\Sch:S\times\Real\times(\RPLUS\times\DIST(S))\mapsto [0,1]$ is measurable if 
$\Sch(s,t,\cdot)\in\DIST(\NEXT(s))$ for all $(s,t)\in S\times\Real$ and 
$\Sch(\cdot,\TR)$ are measurable for all $\TR\in 2^{(\RPLUS\times\DIST(S))}$,
where 
\begin{itemize}
\item $\Sch(s,t,\cdot)$ is a distribution such that $\Sch(s,t,\cdot)(\lambda,\mu) = \Sch(s,t,\lambda,\mu)$, and
\item $\Sch(\cdot,\TR):(S\times\Real)\mapsto [0,1]$ is a function such
  that
 for each $(s,t)\in S\times\Real$, it holds   $\Sch(\cdot,\TR)(s,t)=\sum_{(\lambda,\mu)\in\TR}\Sch(s,t,\lambda,\mu)$.
\end{itemize}

\end{definition}

The schedulers defined in Definition~\ref{def:scheduler} are total time positional,
since they make decisions only based on the current state and total elapsed time,
which are the first and second parameters of $\pi$ respectively.
The third parameter and fourth parameter of $\pi$ denote the rate and the resulting distribution
of the chosen transition respectively.  Given the current state $s$, the total elapsed time $t$, and a
transition $(\lambda,\mu)$, $\pi$ will return the probability with which $(\lambda,\mu)$ will be chosen. 
This is a special case of the general definition of schedulers, which can make decisions
based on the full history, for instance visited states and the sojourn time at each state. 
Given a scheduler $\Sch$, a unique probability measure $\PM_{\Sch,s}$ can be 
determined on the $\sigma$-algebra $\Field_{\Path{\INFPATH}(\MC{C})}$ inductively as below: 
$\PM_{\Sch,s}(C(s_0,I_0,\cdots,s_n),tt)=$
\begin{subnumcases}{}
1 & $n=0 \land s=s_0$\\
0 & $s\neq s_0$\\
\mbox{$\int\limits_{t\in I_0}\sum\limits_{(\lambda,\mu)\in\MI{tr}}\Sch(s_0,tt)(\lambda,\mu)\cdot\mu(s_1)
\cdot\lambda e^{-\lambda t}\cdot\PM_{\Sch,s_1}dt$} & otherwise\label{eq:measure}
\end{subnumcases}
where $\PM_{\Sch,s_1}$ is an abbreviation of $\PM_{\Sch,s_1}(C(s_1,\ldots,s_n),tt+t)$,
$\MI{tr}=\NEXT(s_0)$ and  $tt$ is the parameter denoting the total elapsed time.
One nice property of TTP schedulers is that uniformization does not change time-bounded reachability 
under TTP schedulers~\cite{NeuhausserSK09,RabeS11}.
This result can be extended to cover more properties like $\CSL_{\backslash\!\X}$
and $\CSLstar_{\backslash\!\X}$, which shall be introduced soon.

Besides TTP schedulers, there are other different classes of
schedulers for $\CTMDP$s, some of which are insensitive to
uniformization, whereas some of which may gain or lose information after
uniformization, i.e., properties of a $\CTMDP$ may be changed by
uniformization.  To avoid technical overhead in the presentation, we
refer to~\cite{NeuhausserSK09} for an in-depth discussion of these
different classes of schedulers and their relation to uniformization.

\subsection{Continuous Stochastic Logic}
Logical formulas are important for verification purpose, since they offer a
rigorous and unambiguous way to express properties one may want to
check. Probabilistic computation tree logic ($\PCTL$)~\cite{hansson1994logic}
is often used to express properties of probabilistic systems.
In order to deal with probabilistic systems with exponential sojourn time distributions like $\CTMC$s and $\CTMDP$s,
the continuous stochastic logic ($\CSL$) was introduced to reason about
$\CTMC$s~\cite{Aziz1996VCT,Baier2003MAC}, and recently extended to reason about
$\CTMDP$s in~\cite{MartinK2007}. $\CSL$ contains both
state\footnote{The steady-state operator is omitted in this paper for
  simplicity of presentation. } and path formulas whose syntax is
defined by the following BNFs: 
\begin{align*}
 \phi &::= a\mid\neg\phi\mid\phi\land\phi\mid\MC{P}_{\bowtie p}(\psi), \\
  \psi &::= \XI\phi\mid\phi\UI\phi,
\end{align*}
where $a\in\AP$, $p\in[0,1]$, $\bowtie\ \in\{<,\leq,\geq,>\}$, and $I\subseteq[0,\infty)$ 
is a non-empty closed interval. 

We use $s\models\phi$ to denote that $s$ satisfies the state formula $\phi$, 
while $\omega\models\psi$ denotes that $\omega$ satisfies the path formula $\psi$. 
The satisfaction relation for atomic proposition and Boolean operators is standard. 
Below we give the satisfaction relation for the remaining state and path formulas:
\begin{align*}
s_0\models\MC{P}_{\bowtie p}(\psi) &\text{ iff }\forall\Sch.\PM_{\Sch,s_0}(\{\omega\in\Path{\INFPATH}(\MC{C})\mid\omega\models\psi\})\bowtie p,\\
\omega\models\X^I\phi &\text{ iff }\omega[1]\models\phi\land\Time(\omega,0)\in I,\\
\omega\models\phi_1\UI\phi_2&\text{ iff }\exists i.(\sum_{0\le j< i}\Time(\omega,j)\in I\land\omega[i]\models\phi_2\land(\forall 0\le j<i.\omega[j]\models\phi_1)).
\end{align*}
Intuitively, a state $s_0$ satisfies $\MC{P}_{\bowtie p}(\psi)$ iff no matter how we schedule the transitions of $s_0$ and
its successors, the probability of paths starting from $s_0$ and satisfying $\psi$ is always $\bowtie p$.
This operator has the same semantics as in $\PCTL$. Compared to $\PCTL$, 
the main difference arises in the semantics of the path formulas. 
Given a path $\omega$, we say $\omega\models\X^I\phi$, iff the second state in $\omega$ satisfies $\phi$,
moreover the sojourn time in the first state of $\omega$ is within the time interval $I$.
We say $\omega\models\phi_1\U^I\phi_2$, iff along $\omega$, a state satisfying $\phi_2$ can be reached
at some time point in $I$, and all the preceding states if any satisfy $\phi_1$.
If all time bounds are defined to be equal to $[0,\infty)$, i.e., removing time restrictions, $\CSL$ will degenerate to $\PCTL$. 

Different from~\cite{Baier2003MAC} where the semantics of $\CSL$ is continuous,
in this paper we consider pointwise semantics of $\CSL$. 
This is mainly because the semantics of $\CSLstar$ introduced in Section~\ref{sec:mtl} is also pointwise. 
However, results in Section~\ref{sec:bisimulation} are also valid if we consider continuous semantics.

\paragraph*{Logic Equivalences.}
 Let $\MC{L}$ denote some logic.
 We say that $s$ and $r$ are $\MC{L}$-equivalent, denoted by $s~\sim_{\MC{L}}~r$,
 if they satisfy the same set of $\MC{L}$ state formulas, that is,
 $s\models\phi$ iff $r\models\phi$ for all state formulas $\phi$ in $\MC{L}$,
 similarly for $\sim_{\MC{L}_{\backslash\!\X}}$, where $\MC{L}_{\backslash\!\X}$ denotes the sub-logic 
 of $\MC{L}$ without the $\X^I$ operator. In this paper, $\MC{L}$ will denote 
 either $\CSL$ or $\CSLstar$, which we shall introduce in Section~\ref{sec:mtl}.

\section{Bisimilarity and $\CSL$ Equivalence}\label{sec:bisimulation}
In this section, we first introduce the concept of strong bisimulation for $\CTMDP$s,
which can be seen as a variant of strong bisimulation for $\MDP$s. Then 
we define a sub-class of $\CTMDP$s, called non 2-step recurrent $\CTMDP$s, and 
show that strong bisimulation can be fully characterized by $\CSL$ for non
2-step recurrent $\CTMDP$s. We extend the work to the weak setting and show
similar results for weak bisimulation. Finally, we propose an efficient 
scheme to determine non 2-step recurrent $\CTMDP$s and we show 
that almost all $\CTMDP$ models in practice fall into this class.

\subsection{Strong Bisimulation}\label{sec:strong bisimulation}
The definition of strong bisimulation we shall introduce in this section slightly generalizes
the one introduced in \cite{MartinK2007}. The reason is that we adopt the notion of combined transitions, used in~\cite{SegalaL95} to define \textit{strong probabilistic bisimulation}
for $\PA$s. Combined transitions allow transitions induced by convex combinations of several transitions. 
We shall lift its definition to the setting of $\CTMDP$s.
Let $s\TRANP{\lambda}\mu$ iff there exists $\{s\TRAN{\lambda}\mu_j\}_{j\in J}$ and 
$\{p_j\}_{j\in J}$ such that $\sum_{j\in J}p_j=1$, and $\sum_{j\in J}p_j\cdot\mu_j=\mu$.
The combined transitions of a $\CTMDP$ are almost the same as those for $\PA$s except we need to
take care of the rate of each transition. Here we only allow to combine transitions with the same rate,
otherwise we may change non-trivial properties of a $\CTMDP$, which we will explain soon. 
Below follows the definition of strong bisimulation:

\begin{definition}[Strong Bisimulation]\label{def:strong bisimulation}
Let $\mathcal{R}\subseteq S\times S$ be an equivalence relation. $\MC{R}$ is a
strong bisimulation iff $s~\mathcal{R}~r$ implies that $L(s)=L(r)$ and for each $s\TRAN{\lambda}\mu$, 
there exists $r\TRANP{\lambda}\mu'$ such that $\mu~\MC{R}~\mu'$. 

We write $s~\iBS{}~r$ 
whenever there exists a strong bisimulation $\MC{R}$ such that $s~\MC{R}~r$.
Let strong bisimilarity $\iBS{}$ denote the largest strong bisimulation, 
which is equal to the union of all strong bisimulation relations.
\end{definition}

For $s$ and $r$ to be strong bisimilar, the same set of atomic
propositions should hold at $s$ and $r$.  Furthermore, $s$ should be
able to mimic $r$ stepwise and vice versa, that is, whenever $s$ has a
transition with label $\lambda$ leading to a distribution $\mu$, $r$
should also be able to perform a (combined) transition with the same
label to a distribution $\nu$ such that $\mu$ and $\nu$ match with
each other, i.e., $\mu$ and $\nu$ assign the same probability to each
equivalence class $C\in S/\MC{R}$.  Strong bisimulation defined in
Definition~\ref{def:strong bisimulation} is a conservative extension
of strong probabilistic bisimulation for $\PA$s defined
in~\cite{SegalaL95}, in the sense that it coincides with strong
probabilistic bisimulation if we replace $\lambda$ with actions.

The relation defined above is slightly coarser than the one considered in \cite{MartinK2007}, 
where the combined transition $r\TRANP{\lambda}\mu'$ is replaced by the normal transition $r\TRAN{\lambda}\mu'$. 
In~\cite{MartinK2007}, it was also shown that strong bisimulation is only sound but not complete 
with respect to $\CSL$ equivalence. Even though our definition of strong bisimulation is slightly
coarser, it is still too fine for $\CSL$ equivalence as shown in the following theorem:

\begin{theorem}[\cite{MartinK2007}]\label{thm:martin}
$\iBS{}~\subsetneq~\sim_{\CSL}$.
\end{theorem}
The proof in \cite{MartinK2007} can be directly adapted to 
prove the soundness of our slightly more general strong bisimulation. 
The inclusion in Theorem~\ref{thm:martin} is strict which is illustrated by the
following example:

\begin{figure*}[!t]
\centering
  \begin{tikzpicture}[->,>=stealth,auto,node distance=1.2cm,semithick]
	\tikzstyle{blackdot}=[circle,fill=black,minimum size=6pt,inner sep=0pt]
	\tikzstyle{state}=[minimum size=15pt,circle,draw,thick]
	\tikzstyle{label}=[minimum size=15pt,circle]
	\node[state](sl1){$u_1$};
	\node[state](sl2)[right of=sl1]{$u_2$};  
  	\node[state](sl3)[right of=sl2]{$u_3$};
  	\node[state](sm1)[right of=sl3]{$u_1$};
  	\node[state](sm2)[right of=sm1]{$u_2$};
  	\node[state](sm3)[right of=sm2]{$u_3$};
 	\node[state](sr1)[right of=sm3]{$u_1$};
  	\node[state](sr2)[right of=sr1]{$u_2$};
  	\node[state](sr3)[right of=sr2]{$u_3$};
	\node[blackdot,yshift=0.8cm](d1)[above of=sl2]{}; 
	\node[blackdot,yshift=0.8cm](d2)[above of=sm2]{}; 
	\node[blackdot,yshift=0.8cm](d3)[above of=sr2]{};
	\node[state](r0)[above of=d2]{$r_0$};   
  \path (r0) edge[-]              node[left,yshift=5pt] {1} (d1)
            	   edge[-]              node {1} (d2)
                   edge[-]              node {1} (d3)
        (d1) edge [bend right, dashed,->]  node[left,yshift=5pt] {0.3} (sl1)
               edge [dashed,->] node {0.3} (sl2)
               edge [bend left, dashed,->] node {0.4} (sl3)
        (d2) edge [bend right, dashed,->] node[left,yshift=5pt] {0.4} (sm1)
               edge [dashed,->] node {0.3} (sm2)
               edge [bend left, dashed,->] node {0.3} (sm3)
		(d3) edge [bend right, dashed,->] node[left,yshift=5pt] {0.5} (sr1)
               edge [dashed,->] node {0.4} (sr2)
               edge [bend left, dashed,->] node {0.1} (sr3); 
     \node[label](b)[below of=sm2,yshift=0.3cm]{(b)};
     \node[label](a)[above of=r0, node distance=1cm]{(a)};
     \node[state](s0)[above of=r0, node distance=5cm]{$s_0$};
     \node[blackdot](d4)[above of=d1, node distance=5cm]{};
      \node[blackdot](d5)[above of=d3, node distance=5cm]{};
     \node[state](rl1)[above of=sl1, node distance=5cm]{$u_1$};
     \node[state](rl2)[right of=rl1]{$u_2$};
     \node[state](rl3)[right of=rl2]{$u_3$};
     \node[state](rr2)[above of=sr2, node distance=5cm]{$u_2$};
     \node[state](rr1)[left of=rr2]{$u_1$};
     \node[state](rr3)[right of=rr2]{$u_3$};        
     \path (s0) edge[-]              node[left,yshift=5pt] {1} (d4)
            	   edge[-]              node {1} (d5)
        (d4) edge [bend right, dashed,->]  node[left,yshift=5pt] {0.3} (rl1)
               edge [dashed,->] node {0.3} (rl2)
               edge [bend left, dashed,->] node {0.4} (rl3)
        (d5) edge [bend right, dashed,->] node[left,yshift=5pt] {0.5} (rr1)
               edge [dashed,->] node {0.4} (rr2)
               edge [bend left, dashed,->] node {0.1} (rr3); 
\end{tikzpicture}
  \caption{\label{fig:counterexample completeness}%
 Counterexample of the completeness of strong bisimulation.}
\end{figure*}
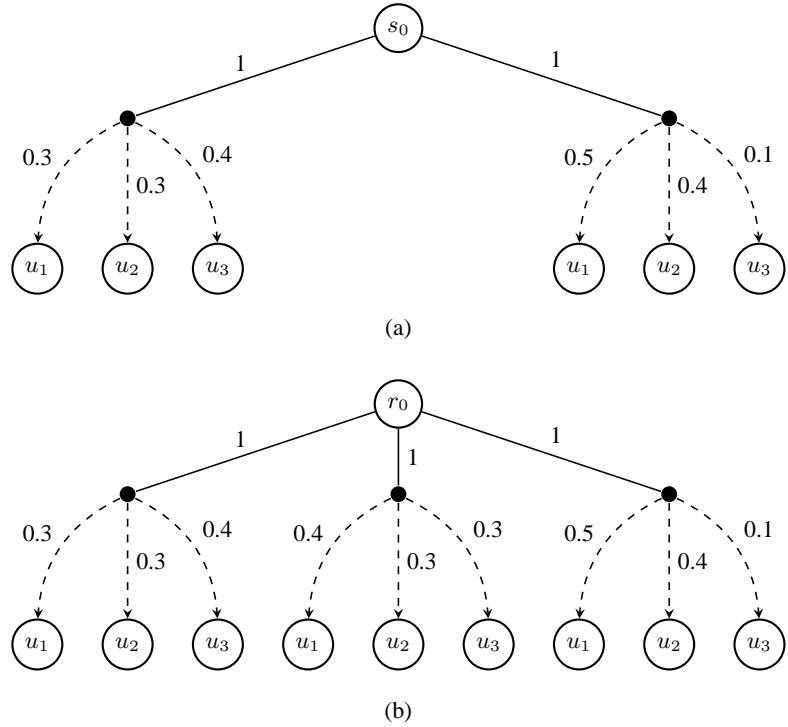

\begin{example}\label{ex:counterexample completeness}
Suppose we are given two states $s_0$ and $r_0$ of a $\CTMDP$ depicted in Fig.~\ref{fig:counterexample completeness}~(a)
and (b) respectively, where all states have different atomic propositions except $L(s_0)=L(r_0)$.
Assume $u_i$ are silent for $i=1,2,3$, our aim is to show that $s_0$ and $r_0$ 
satisfy the same set of $\CSL$ formulas, while they are not strong bisimilar by Definition~\ref{def:strong bisimulation}.

We first show that $s_0~\sim_{\CSL}~r_0$, i.e., $s_0\models\phi$ implies 
$r_0\models\phi$ for any $\phi$ and vice versa. The only non-trivial cases are
the time-bounded reachabilities from $s_0$ and $r_0$ to states in $C\subseteq\{u_1,u_2,u_3\}$. For instance
the maximal probability from $s_0$ and $r_0$ to $\{u_2,u_3\}$ in time interval $[a,b]$ is equal to
 $0.7\cdot(e^{-a}-e^{-b})$, irrelevant of the middle transition of $r_0$. Similarly, we can check that
for other $C$, the maximal (or minimal) probabilities from $s_0$ and $r_0$ to $C$ in time interval $I$
are all independent from the middle transition of $r_0$. Therefore we conclude that $s_0~\sim_{\CSL}~r_0$.

Secondly, we show that it does not hold that $s_0~\sim~r_0$ according to Definition~\ref{def:strong bisimulation}. 
We prove by contradiction. Assume that there exists a strong bisimulation $\MC{R}$ such that $s_0~\MC{R}~r_0$.
By Definition~\ref{def:strong bisimulation}, for the middle transition of $r_0$, i.e., $r_0\TRAN{1}\mu'$ where
$\mu'(u_1)=0.4,\mu'(u_2)=0.3,$ and $\mu'(u_3)=0.3$, we need to find a transition $s_0\TRANP{1}\mu$ of $s_0$
such that $\mu~\MC{R}~\mu'$. Since $u_1,u_2,$ and $u_3$ have different atomic propositions, $(u_i,u_j)\not\in\MC{R}$
for any $1\le i\neq j\le 3$. Therefore the only possibility is that $\mu(u_1)=0.4,\mu(u_2)=0.3,$ and $\mu(u_3)=0.3$. However
that is impossible, such $\mu$ cannot be the resulting distribution of any (combined) transition of $s_0$.
Otherwise there would exist $w_1,w_2>0$ such that $w_1+w_2=1$, $0.3\cdot w_1 + 0.5\cdot w_2=0.4$, and 
$0.3\cdot w_1+0.4\cdot w_2=0.3$ according to the definition of combined transition, which
is clearly not possible. Hence we conclude that $s_0~\not\sim~r_0$, and $\sim$ is finer than
$\sim_{\CSL}$.\qed
\end{example}

In~\cite{RabeS11} randomized schedulers allow to combine transitions with
different rates, i.e., the combined transition is defined as: 
$s\TRANP{\lambda}\mu$ iff there exist 
$\{s\TRAN{\lambda_i}\mu_i\}_{i\in I}$ and $\{p_i\}_{i\in I}$ such that  
$\sum_{i\in I}p_i\cdot\lambda_i=\lambda$  and $\sum_{i\in I}p_i\cdot\mu_i=\mu$, 
where $p_i\in[0,1]$  for each $i\in I$ and $\sum_{i\in I}p_i=1$. 
By adopting this definition of combined transition in Definition~\ref{def:strong bisimulation},
we will obtain a coarser strong bisimulation. However it turns out that
this new definition of strong bisimulation is too coarse for $\CSL$ equivalence, 
since there exist two states which are strong bisimilar according to
the new definition, but they satisfy different $\CSL$ formulas. Refer
to the following example:

\iffalse
\begin{figure*}[!t]
\centering
  \begin{tikzpicture}[->,>=stealth,auto,node distance=2cm,semithick]
	\tikzstyle{state}=[minimum size=15pt,circle,draw,thick]
	\tikzstyle{label}=[minimum size=15pt,circle]
        \tikzstyle{invisiblestate}=[]
     \node[state](s11){$u_1$};
     \node[invisiblestate](s12)[right of=s11]{};
     \node[state](s13)[right of=s11, node distance=4cm]{$u_1$};
     \node[state](s0)[above of=s12]{$s_1$};
     \node[state](s21)[right of=s13]{$u_1$};
     \node[state](r0)[right of=s0,node distance=6cm]{$r_1$};
     \node[state](s22)[right of=s21]{$u_1$};
     \node[state](s23)[right of=s22]{$u_1$};
     \path (s0) edge              node[left,yshift=5pt] {1} (s11)
                edge              node {4} (s13)
           (r0) edge              node[left,yshift=5pt] {1} (s21)
            	edge              node {2} (s22)
                edge              node {4} (s23);              
\end{tikzpicture}
  \caption{\label{fig:combined transition}
 Transitions with different rates cannot be combined.}
\end{figure*}
\fi

\begin{example}\label{ex:2-step recurrent ex2}
Suppose that we have two states $s_1$ and $r_1$ %depicted as in Fig.~\ref{fig:combined transition},
such that $s_1$ has two non-deterministic transitions which can evolve
into $u_1$ with rates 1 or 4 respectively. The state $r_1$ is the same as $s_1$ except
that it can evolve into $u_1$ with an extra transition of rate $2$. 
Also we assume that $L(s_1)=L(r_1)$ and $u_1$
is a silent state with $L(u_1)\not\subseteq L(s_1)$.
Suppose that we adopt the new definition of combined transition in Definition~\ref{def:strong bisimulation}
by allowing to combine transitions with different rates,
we shall show that $s_1$ and $r_1$ are strong bisimilar, but 
they are not $\CSL$-equivalent.

We first show that $s_1$ and $r_1$ are strong bisimilar. Let $\MC{R}$ be an equivalence relation only
equating $s_1$ and $r_1$,
it suffices to prove that $\MC{R}$ is a strong bisimulation. The only non-trivial
case is when $r_1\TRAN{2}\DIRAC{u_1}$, we need to find a matching transition of $s_1$.
Since we allow to combine transitions of different rates, a combined transition
$s_1\TRANP{2}\DIRAC{u_1}$ can be obtained by assigning weights $\frac{2}{3}$ and $\frac{1}{3}$
to transitions $s_1\TRAN{1}\DIRAC{u_1}$ and $s_1\TRAN{4}\DIRAC{u_1}$ respectively.
Therefore we conclude that $s_1$ and $r_1$ are strong bisimilar.

Secondly, we show that $s_1$ and $r_1$ are not $\CSL$ equivalent. It suffices
to find a formula $\phi$ such that $s_1\models\phi$ but $r_1\not\models\phi$.
Let $\psi=\X^{[a,b]} L(u_1)$ where $0\leq a < b$.
The probabilities for paths starting from $s_1$ and satisfying $\psi$
by choosing the transitions with rates 1, 2, and 4 are equal to
$e^{-a}-e^{-b}$, $e^{-2a}-e^{-2b}$, and $e^{-4a}-e^{-4b}$ respectively.
We need only to find $a$ and $b$ such that
$e^{-2a}-e^{-2b} > \max\{e^{-a}-e^{-b}, e^{-4a}-e^{-4b}\}$.
Let $a=0.2$ and $b=1$, then $e^{-a}-e^{-b}\approx0.45$, $e^{-2a}-e^{-2b}\approx0.53$, 
and $e^{-4a}-e^{-4b}\approx 0.43$. Let $\phi=\MC{P}_{\leq 0.46}(\X^{[0.2,1]}L(u_1))$,
obviously $s_1\models\phi$, but $r_1\not\models\phi$, which means that
$s_1$ and $r_1$ are not $\CSL$-equivalent.\qed
\end{example}

Example~\ref{ex:2-step recurrent ex2} also shows that in order for
two states satisfying the same $\CSL$ formulas,
it is necessary for them to have transitions with the same exit rates,
otherwise we can always find $\CSL$ formulas distinguishing them,
which also justifies that we only allow to combine transitions with the same rate
in Definition~\ref{def:strong bisimulation}.

We have shown in 
Example~\ref{ex:counterexample completeness} that $\iBS{}$ is not complete
with respect to $\sim_{\CSL}$. However in the sequel we shall identify a special class
of $\CTMDP$s, in which the completeness holds. 
We first give two examples for inspiration:

\begin{example}\label{ex:successors}
  In this example, we show that,
  it is impossible to construct similar states as $s_0$ and $r_0$ in 
  Example~\ref{ex:counterexample completeness} such that they  are not strong bisimilar but only have
  2 distinct successors. 
  
  Let $s_2$ and $r_2$ denote the two states depicted in Fig.~\ref{fig:successors}, where $x\in[0,1]$ denotes
  an arbitrary or unknown probability and all states
  have different atomic propositions except that $L(s_2)=L(r_2)$. Our aim is to show
  that states in form of $s_2$ and $r_2$ must be strong bisimilar, provided that $s_2~\sim_{\CSL}~r_2$.
  First we show that $x\in[\frac{1}{4},\frac{1}{2}]$ in order that $s_2~\sim_{\CSL}~r_2$.
  This is done by contradiction. Assume that $x>\frac{1}{2}$ and let $\psi=\X^{[0,\infty)}(L(u_1))$.
  Then the maximal probability of paths starting from $s_2$ and satisfying $\psi$ is equal to $\frac{1}{2}$,
  while the maximal probability of paths starting from $r_2$ and satisfying $\psi$ is equal to $x$.
  Since $x>\frac{1}{2}$, $s_2\models\MC{P}_{\le \frac{1}{2}}(\psi)$, while $r_2\not\models\MC{P}_{\le\frac{1}{2}}(\psi)$,
  therefore $s_2~\not\sim_{\CSL}~r_2$. Similarly, we can show that it is not possible for $x<\frac{1}{4}$, hence
  it holds that $x\in[\frac{1}{4},\frac{1}{2}]$.
  
  Secondly, we show that $s_2~\sim~r_2$ given that $x\in[\frac{1}{4},\frac{1}{2}]$. Let
  $\MC{R}$ be an equivalence relation only equating $s_2$ and $r_2$, it suffices to show that $\MC{R}$ is
  a strong bisimulation according to Definition~\ref{def:strong bisimulation}.
  Let $\mu_1,\mu_2,$ and $\mu_3$ be distributions defined in Fig.~\ref{fig:successors}.
  The only non-trivial case is when $r_2\TRAN{1}\mu_2$,
  we need to show that there exists $w_1$ and $w_2$ such that $w_1+w_2=1$, $(w_1\cdot\mu_1+w_2\cdot\mu_3)~\MC{R}~\mu_2$.
  Let $w_1=2-4x$ and $w_2=4x-1$, it is easy to verify that $w_1,w_2\in[0,1]$ and $w_1+w_2=1$, since $x\in[\frac{1}{4},\frac{1}{2}]$.
  Moreover, $w_1\cdot\mu_1 + w_2\cdot\mu_3=\mu_2$, since $w_1\cdot\frac{1}{4}+w_2\cdot\frac{1}{2}=x$ and $w_1\cdot\frac{3}{4}+w_2\cdot\frac{1}{2}=1-x$.
  Therefore $s_2\TRANP{1}\mu_2$ as desired, and $\MC{R}$ is indeed a strong bisimulation.    \qed
\end{example}

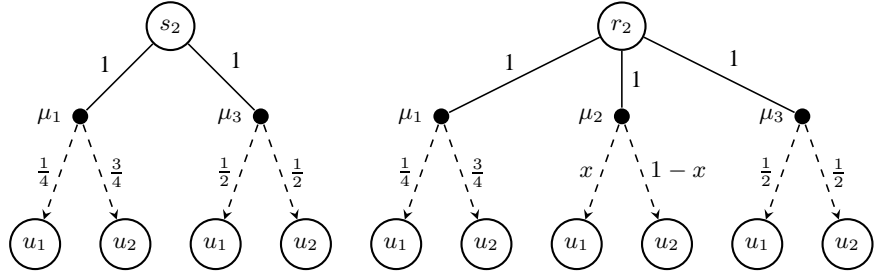
\begin{figure*}[!t]
\centering
  \begin{tikzpicture}[->,>=stealth,auto,node distance=1.2cm,semithick]
	\tikzstyle{state}=[minimum size=15pt,circle,draw,thick]
	every label/.style=draw
        \tikzstyle{blackdot}=[circle,fill=black, minimum size=6pt,inner sep=0pt]
     \node[state](u11){$u_1$};
     \node[state](u21)[right of=u11]{$u_2$};
     \node[state](u12)[right of=u21]{$u_1$};
     \node[state](u22)[right of=u12]{$u_2$};
     \node[state](u13)[right of=u22]{$u_1$};
     \node[state](u23)[right of=u13]{$u_2$};
     \node[state](u14)[right of=u23]{$u_1$};
     \node[state](u24)[right of=u14]{$u_2$};
     \node[state](u15)[right of=u24]{$u_1$};
     \node[state](u25)[right of=u15]{$u_2$};
     \node[blackdot,label={[label distance=0pt]180:{$\mu_1$}}](d1)[above of=u11,xshift=0.6cm,yshift=0.5cm]{};
     \node[blackdot,label={[label distance=0pt]180:{$\mu_3$}}](d2)[above of=u12,,xshift=0.6cm,yshift=0.5cm]{};
     \node[blackdot,label={[label distance=0pt]180:{$\mu_1$}}](d3)[above of=u13,,xshift=0.6cm,yshift=0.5cm]{};
     \node[blackdot,label={[label distance=0pt]180:{$\mu_2$}}](d4)[above of=u14,,xshift=0.6cm,yshift=0.5cm]{};
     \node[blackdot,label={[label distance=0pt]180:{$\mu_3$}}](d5)[above of=u15,,xshift=0.6cm,yshift=0.5cm]{};
     \node[state](s2)[above of=d1,xshift=1.2cm]{$s_2$};
     \node[state](r2)[above of=d4]{$r_2$};
     \path (s2) edge[-]              node[left,yshift=5pt] {1} (d1)
                edge[-]              node {1} (d2)
           (r2) edge[-]              node[left,yshift=5pt] {1} (d3)
            	edge[-]              node {1} (d4)
                edge[-]              node {1} (d5)
           (d1) edge[dashed]         node[left] {$\frac{1}{4}$} (u11)
                edge[dashed]         node[yshift=-9pt] {$\frac{3}{4}$} (u21)
           (d2) edge[dashed]         node[left] {$\frac{1}{2}$} (u12)
                edge[dashed]         node[yshift=-9pt] {$\frac{1}{2}$} (u22)
           (d3) edge[dashed]         node[left] {$\frac{1}{4}$} (u13)
                edge[dashed]         node[yshift=-9pt] {$\frac{3}{4}$} (u23)
           (d4) edge[dashed]         node[left] {$x$} (u14)
                edge[dashed]         node[yshift=-6pt] {$1-x$} (u24)
           (d5) edge[dashed]         node[left] {$\frac{1}{2}$} (u15)
                edge[dashed]         node[yshift=-9pt] {$\frac{1}{2}$} (u25);
  \end{tikzpicture}
  \caption{\label{fig:successors}
 $s_2$ can always simulate the middle transition of $r_2$, as long as $\frac{1}{4}\le x\le\frac{1}{2}.$}
\end{figure*}

In order for Example~\ref{ex:counterexample completeness} being a valid 
counterexample for $\sim_{\CSL}~\subseteq~\sim$, we have made another assumption that $u_i$ ($i=1,2,3$)
are silent, i.e., they cannot evolve into other states not equivalent to themselves with positive probability.
This assumption is also crucial which can be seen by the following example:
\begin{example}\label{ex:2-step recurrent ex1}
Consider again the two states $s_0$ and $r_0$ introduced in Example~\ref{ex:counterexample completeness},
where we prove that $s_0$ and $r_0$ are $\CSL$ equivalent.
Now suppose that $u_3$ is not silent, but can evolve into some state $u'_3$ with rate 1, where $u'_3$
is a state with different atomic propositions from all the others. We are going to show
that $s_0$ and $r_0$ are not $\CSL$ equivalent anymore with this slight change.
Consider the path formula: 
$\psi=(L(s_0)\lor L(u_3)) \U^{[0,b]} (L(u_2)\lor L(u'_3))$, we can show that the 
probabilities of paths starting from $r_0$ and satisfying $\psi$ by choosing the left, middle, and right transitions
are equal to:
$L = 0.3\cdot w_1 + 0.4\cdot w_2,$
$M = 0.3\cdot w_1 + 0.3\cdot w_2,$ and
$R = 0.4\cdot w_1 + 0.1\cdot w_2$
respectively, where $w_1 = 1 - e^{-b}$ and $w_2 = 1 - e^{-b} - b\cdot e^{-b}$.
It suffices to find a $b$ such that $M < \min\{L,R\}$,
which means that the middle transition of $r_0$ dominates the minimal probability of
satisfying $\psi$. Such $b$ exists, for instance, by letting $b=1$ 
we obtain: $L \approx 0.295$, $M \approx 0.269$, and $R\approx 0.279$, apparently, $M < \min\{L,R\}$.
In other words, let $b=1$ in $\psi$, we have $s_0\models\MC{P}_{\geq R}(\psi)$, but
$r_0\not\models\MC{P}_{\geq R}(\psi)$, since there exists a scheduler of $r_0$, i.e., the one choosing 
the middle transition of $r_0$ such that the probability of satisfying $\psi$ is equal to $M$, which
is strictly less than $R$. Therefore $s_0~\not\sim_{\CSL}~r_0$. \qed
\end{example}

In Example~\ref{ex:counterexample completeness}, we have shown that 
$s_0$ and $r_0$ satisfy the same $\CSL$ formulas, but they are not strong bisimilar.
However in Examples~\ref{ex:successors} and \ref{ex:2-step recurrent ex1}, 
we show that without the two assumptions:
\begin{itemize}
\item $s_0$ and $r_0$ should have more than 2 states among their successors;
 \item there exists no successor which can evolve into a state
 not $\CSL$ equivalent to other states with positive probability, 
\end{itemize}
we can guarantee that either $s_0$ and $r_0$ are strong bisimilar,
or they are not $\CSL$ equivalent.
These intuitions lead us to the special class of $\CTMDP$s, which we call
\emph{non 2-step recurrent} $\CTMDP$s in the sequel.
\begin{definition}[2-step Recurrent]\label{def:2-step recurrent}
Let $\MC{R}$ be an equivalence relation on $S$.
A state $s$ is said to be
\emph{2-step recurrent} with respect to $\MC{R}$ iff 
$s$ is not silent, $\ABS{\SUC(s)}>2$, and 
\begin{equation}\label{eq:2-step recurrent}\tag{r1}
\exists(s\TRAN{\lambda}\mu).(\forall s'\in(\SUPP(\mu)\setminus\ESET[\MC{R}]{s}).\forall(s'\TRAN{\lambda'}\nu).\nu(C)=1),
\end{equation}
where $C=(\ESET[\MC{R}]{s}\cup\ESET[\MC{R}]{s'})$.

We say $\MC{C}$ is \emph{2-step recurrent} with respect to $\MC{R}$,
iff there exists $s\in S$ such that $s$ is 2-step recurrent with respect to $\MC{R}$, 
otherwise it is non 2-step recurrent with respect to $\MC{R}$.
Moreover, we say that $s$ (or $\MC{C}$) is
(non) 2-step recurrent iff it is (non) 2-step recurrent with respect to $\sim_{\CSL}$.
\end{definition}
In other words, for a state $s$ to be 2-step recurrent,
it must be not silent and have more than 2 successors.
Remind that each silent state can be replaced by a single state without
changing properties of a $\CTMDP$. After doing so, each silent state
will only have one successor which is itself, so the requirement
of non silence can be subsumed by $\ABS{\SUC(s)}>2$ in this case. 
Let us explain the more involved condition given in Eq.~(\ref{eq:2-step recurrent}).
Eq.~(\ref{eq:2-step recurrent}) says that a 2-step recurrent state $s$ must also satisfy:
There exists $s\TRAN{\lambda}\mu$ such that for all states in  $\SUPP(\mu)$ 
except those in $\ESET[\MC{R}]{s}$, they can only
evolve into states equivalent to $s$ or themselves.

\begin{example}\label{ex:2-step recurrent CTMDPs}
  We show some examples of (non) 2-step recurrent states. 
  First of all, states $s_0$ and $r_0$ in Example~\ref{ex:counterexample completeness}
  are 2-step recurrent, since they are not silent and have more than 2 successors. 
  Moreover all successors $u_i$ ($i=1,2,3$) are silent, i.e., can only evolve into
  states which are $\CSL$ equivalent to themselves. 
  However if we add an extra transition to $u_3$ as in Example~\ref{ex:2-step recurrent ex1},
  $s_0$ will be non 2-step recurrent, since $u_3$ can reach the state $u'_3$
  with probability 1, where $u'_3$ is not $\CSL$ equivalent to either $u_3$ or $s_0$.
  For similar reasons, $r_0$ is also non 2-step recurrent.
  
  Secondly, States $s_1$ and $r_1$ in Example~\ref{ex:2-step recurrent ex2} and $s_2$ and $r_2$
  in Example~\ref{ex:successors} are trivially non 2-step recurrent, since
  the number of their successors is $\le 2$. \qed
\end{example}

Definition~\ref{def:2-step recurrent} seems tricky, however,
we shall show that there exists an efficient scheme to check whether a given $\CTMDP$
is 2-step recurrent or not. More importantly, we shall see later in Remark~\ref{remark:2-step recurrent} 
that the class of non 2-step recurrent $\CTMDP$s contains an important part of $\CTMDP$ models, 
in particular those found in practice.

Now we are ready to show the main contribution of this paper. By restricting to the set of 
non 2-step recurrent $\CTMDP$s, we are able to prove that the classical strong bisimulation
defined in Definition~\ref{def:strong bisimulation} is both sound and complete with respect to the
$\CSL$ equivalence, which is formalized in the following theorem.
\begin{theorem}\label{thm:equivalent strong}
If $\MC{C}$ is non 2-step recurrent, $\sim~=~\sim_{\CSL}$.
\end{theorem}

\subsection{Weak Bisimulation}\label{sec:weak bisimulation}
In this section we will introduce a novel notion of \emph{weak
  bisimulation} for $\CTMDP$s. Our definition of weak
bisimulation is directly motivated by the well-known fact 
that uniformization does not alter time-bounded reachabilities for $\CTMDP$s~\cite{NeuhausserSK09,RabeS11}
when TTP schedulers are considered. 
Similar as in Section~\ref{sec:strong bisimulation}, we also show that weak bisimulation is both sound and complete for 
$\CSL_{\backslash\!\X}$ over non 2-step recurrent $\CTMDP$s. 
%The section ends up with a discussion about
%why the results do not hold for general $\CTMDP$s, and motivates the
%study of a sequence of bisimulations in the next section. 
We shall introduce the definition of weak bisimulation first.
\begin{definition}[Weak bisimulation]\label{def:weak bisimulation}
We say that states $s$ and $r$ in $\MC{C}$ are weak bisimilar, denoted by $s~\WBS~r$,
whenever $\UNIFORM{s}~\sim~\UNIFORM{r}$ in the uniformized \emph{$\CTMDP$} $\UNIFORM{\MC{C}}$.
\end{definition}

The way we define weak bisimulation here is different from the definition of
weak bisimulation for $\CTMC$s in~\cite{BaierKHW05}, where a conditional measure is
considered, see Definition~\ref{def:weak bisimulation CTMC} for the detailed definition. Moreover
we will show in Section~\ref{sec:relation to ctmc} that for $\CTMC$s our weak bisimulation coincides with
weak bisimulation defined in~\cite{BaierKHW05}. Even though the resulting uniformized $\CTMDP$ depends
on the chosen rate $E$ as shown in Definition~\ref{def:uniformization}, it is worth mentioning that
weak bisimulation given in Definition~\ref{def:weak bisimulation} is independent of $E$. 
Since if two states are strong bisimilar in
a uniformized $\CTMDP$, they will be strong bisimilar in any uniformized $\CTMDP$ no matter
which value we choose for $E$.

The following lemma establishes some properties:
\begin{lemma}\label{lem:sim_properties}\  
 \begin{enumerate}
  \item $\sim~\subseteq~\WBS$,
  \item for uniformized \emph{$\CTMDP$}s, $\sim~=~\WBS$.
  \end{enumerate}
\end{lemma}

As we mentioned above, by uniformizing a $\CTMDP$ we will not change its satisfiability of $\CSL_{\backslash\!\X}$
provided that only TTP schedulers are considered. 
Therefore we have the following lemma saying that if two states
satisfy the same formulas in $\CSL_{\backslash\!\X}$, then they will satisfy the same formulas in
$\CSL$ after uniformization and vice versa.

\begin{lemma}\label{lem:CSL preserved}
$s~\sim_{\CSL_{\backslash\!\X}}~r$ in $\MC{C}$ iff $\UNIFORM{s}~\sim_{\CSL}~\UNIFORM{r}$ in $\UNIFORM{\MC{C}}$.
\end{lemma}

The following theorem says that our weak bisimulation is sound for $\sim_{\CSL_{\backslash\!\X}}$,
and particularly when the given $\CTMDP$ is non 2-step recurrent, weak bisimulation can be used
to fully characterize $\CSL_{\backslash\!\X}$ equivalence.
\begin{theorem}\label{thm:equivalent weak}
$\WBS~\subsetneq~\sim_{\CSL_{\backslash\!\X}}$. If $\UNIFORM{\MC{C}}$ is non 2-step recurrent,
$\WBS~=~\sim_{\CSL_{\backslash\!\X}}$.
\end{theorem}

Theorem~\ref{thm:equivalent weak} works if we restrict to only TTP schedulers. However, this is not a 
restriction. Since it has been proved in~\cite{RabeS11,Buchholz2011MCA} that
there always exists an optimal scheduler in TTP for any path property in $\CSL_{\backslash\!\X}$.

\subsection{Determining 2-step Recurrent $\CTMDP$s}\label{sec:determining 2-step}
In Theorem~\ref{thm:equivalent strong} and \ref{thm:equivalent weak}, the completeness holds only for
$\CTMDP$s which are non 2-step recurrent. Hence it is important that 2-step recurrent $\CTMDP$s
can be checked efficiently. This section discusses a simple
procedure for determining (non) 2-step recurrent $\CTMDP$s. 
Before presenting the decision scheme, we shall introduce
the following lemma, which holds by applying the
definition of 2-step recurrent $\CTMDP$s directly:
\begin{lemma}\label{lem:2-step recurrent}
Given two equivalence relations $\MC{R}$ and $\MC{R}'$ over $S$ such that $\MC{R}\subseteq\MC{R}'$,  
if $\MC{C}$ is 2-step recurrent with respect to $\MC{R}$, then it is 2-step recurrent with respect to 
$\MC{R}'$, or equivalently if $\MC{C}$ is non 2-step recurrent with respect to $\MC{R}'$, then it is non 2-step 
recurrent with respect to $\MC{R}$.
\end{lemma}

Lemma~\ref{lem:2-step recurrent} suggests a simple way to check
whether a given $\CTMDP$ $\MC{C}$ is 2-step recurrent. Given an arbitrary equivalence relation $\MC{R}$ such that
$\sim~\subseteq~\sim_{\CSL}~\subseteq~\MC{R}$, by Lemma~\ref{lem:2-step
  recurrent}, we can first check whether $\MC{C}$ is 2-step recurrent
with respect to $\MC{R}$. Proper candidates for $\MC{R}$ should be as
fine as possible, but also can be determined efficiently. 
For instance, we can let $\MC{R}=\{(s,r)\mid L(s)=L(r)\}$, or a finer  
equivalence relation defined as follows: 
$s~\MC{R}~r$ iff for each $C\in S/\MC{R}$ and $s\TRAN{\lambda}\mu$,
there exists $r\TRAN{\lambda}\mu'$ such that $\mu'(C)\ge\mu(C)$. 
Such $\MC{R}$ is coarser than $\sim_{\CSL}$, and can be computed efficiently 
in polynomial time.

If $\MC{C}$ is not 2-step recurrent
with respect to $\MC{R}$, we know that $\MC{C}$ is non 2-step
recurrent with respect to $\sim_{\CSL}$ either.  Otherwise we
continue to check whether $\MC{C}$ is 2-step recurrent with respect to $\sim$, 
if the answer is yes, then $\MC{C}$ is 2-step recurrent with respect to $\sim_{\CSL}$ too. 
Note that $\sim$ can also be
computed in polynomial time, see \cite{ZhangHEJ08} for details.  In the
remaining cases, namely when $\MC{C}$ is 2-step recurrent
with respect to $\MC{R}$, but not for $\sim$, we cannot conclude anything,
instead the relation $\sim_{\CSL}$ shall be computed
first for a definite answer.

As we discussed above, sometimes we need to use $\sim_{\CSL}$ to decide whether
a given $\CTMDP$ is 2-step recurrent or not. But it turns out that
$\sim_{\CSL}$ is hard to compute in general.
Actually, we can prove the following lemma showing that the decision of $\sim_{\CSL}$ and $\sim_{\CSL_{\backslash\!\X}}$ is NP-hard.
\begin{lemma}\label{lem:NP result}
It is NP-hard to decide whether $s~\sim_{\CSL}~r$ and $s~\sim_{\CSL_{\backslash\!\X}}~r$.
\end{lemma}

\begin{remark}\label{remark:2-step recurrent}
We have implemented the above described scheme to check whether some models in practice are 2-step 
recurrent or not. Even though the implemented classification scheme is not complete since we 
do not compute $\CSL$ equivalence, it has been shown quite useful in practice.
Our initial experiments show that the
non 2-step recurrent $\CTMDP$s consist of most models in practice. For instance the models of ``Erlang 
Stages"~\cite{ZhangN10},
``Stochastic Job Scheduling"~\cite{BrunoDF81}, ``Fault-Tolerant Work Station 
Cluster"~\cite{HaverkortHK00,KatoenZHHJ09}, and
``European Train Control System"~\cite{BoedeHHJPPWB06} are all non 2-step recurrent, which means that 
strong bisimulation coincides with $\sim_{\CSL}$ on these models. To be more confident, we also checked $\MDP$ models from
the PRISM~\cite{KNP11} benchmark interpreted as $\CTMDP$ models by interpreting all probabilities as rates.
We found that all of them are non 2-step recurrent. \qed
\end{remark}

\section{Bisimilarity  and  $\CSLstar$ Equivalence}
\label{sec:mtl}
In this section we study the relation between bisimilarity and $\CSLstar$ equivalence.
We first introduce $\CSLstar$, then show that strong bisimulation can be
fully characterized by $\CSLstar$ for arbitrary $\CTMDP$s. Then we 
extend the work to weak bisimulation.

\subsection{$\CSLstar$}
As $\CTL^*$ and $\PCTL^*$ can be seen as extensions of $\CTL$ and $\PCTL$ respectively, 
$\CSLstar$ can also be seen as an extension of $\CSL$, where the path formula is defined by the Metric Temporal Logic 
($\MTL$)~\cite{Koymans1990SRP}. $\MTL$ extends linear temporal logic~\cite{Pnueli1977TLP} by associating 
each temporal operator with a time interval. It is a popular logic used to specify
properties of real-time systems and has been extensively studied 
in the literature~\cite{Alur1994RTL,Ouaknine2005DMT,bouyer2007cost,jenkins2010alternating}.
The logic $\MTL$ was also extended to $\CTMC$s in~\cite{Chen2011TVC},
where the authors studied the problem of model checking $\CTMC$s 
against $\MTL$ specifications. 
Formally, the syntax of $\CSLstar$ is defined by the following BNFs:
\begin{align*}
  \phi ::=& a\mid \neg\phi \mid \phi\land\phi \mid \MC{P}_{\bowtie p}(\psi), \\
  \psi ::=& \phi \mid \neg\psi \mid \psi\land\psi \mid \X^I\psi \mid \psi\U^I\psi. 
\end{align*}
The semantics of state formulas is the same as $\CSL$, while the semantics of
path formulas is more involved, since we may have different and embedded time bounds.
As for \MTL{}, there are two different semantics for the path formulas:
continuous semantics and pointwise semantics. These two semantics make non-trivial
differences in real-time systems, see~\cite{Ouaknine2005DMT} for details. 
We shall focus on the pointwise semantics as for $\CSL$ in this paper. 
Given a path $\omega$ and  a path formula $\psi$ of $\CSLstar$, 
the satisfiability $\omega\models\psi$ is defined inductively as follows:
$\omega\models a$ iff $a\in L(\omega[0])$, 
$\omega\models \neg\psi$ iff $\omega\not\models\psi$,
$\omega\models \psi_1\land\psi_2$ iff
$\omega\models\psi_1\land\omega\models\psi_2$,
$\omega\models\X^I\psi$ iff $\SUF{1}\models\psi\land\Time(\omega,0)\in I$, and
$$
    \omega\models \psi_1\U^I\psi_2 
    \text{ iff }\exists i.(\SUF{i}\models\psi_2\land\sum_{0\le j<i}\Time(\omega,j)\in I\land
    (\forall 0\le j<i.\SUF{j}\models\psi_1)).
$$

\subsection{Strong Bisimulation}\label{sec:mtl strong}
In this section we prove the soundness and completeness of strong bisimulation
with respect to $\CSLstar$ equivalence. Different from $\CTL$ and its extension $\CTL^*$,
whose equivalences coincide on labelled transition systems~\cite{Browne1988CFK}, 
the extension from $\CSL$ to $\CSLstar$ is non-trivial, as we shall show in this section that $\CSLstar$
can fully characterize strong bisimulation for arbitrary $\CTMDP$s.  
We reconsider Example~\ref{ex:counterexample completeness} for inspiration:
\begin{example}\label{ex:csl 1}
  Let $s_0$ and $r_0$ be the states introduced in
  Example~\ref{ex:counterexample completeness}, where we have shown that
  $s_0$ and $r_0$ are not bisimilar, but satisfy the same $\CSL$
  formula. However if we consider $\CSLstar$, $s_0$ and $r_0$ are
  not $\CSLstar$ equivalent. It suffices to find a formula $\phi$ in $\CSLstar$
  such that $s_0\models\phi$, but $r_0\not\models\phi$.
  Let $\psi :=
  (L(s_0)\UI[{[0.6,\infty)}]L(u_1))\lor (L(s_0)\UI[{[1,\infty)}]
  L(u_3)),$ 
  then the maximal probability of paths starting from $s_0$
  and satisfying $\psi$ is equal to
  $\max\{0.3\cdot e^{-0.6} + 0.4\cdot e^{-1},0.5\cdot e^{-0.6} + 0.1\cdot e^{-1}\}<0.312,$
  while the probability for $r_0$ is equal to 
  $\max\{0.3\cdot e^{-0.6} + 0.4\cdot e^{-1},0.4\cdot e^{-0.6} + 0.3\cdot e^{-1},0.5\cdot e^{-0.6} + 0.1\cdot e^{-1}\}>0.312,$
  thus $s_0\models\MC{P}_{\le 0.312}(\psi)$, while $r_0\not\models\MC{P}_{\le 0.312}(\psi)$, which 
  indicates $s_0~\not\sim_{\CSLstar}~r_0$. Note $\psi$ is not a valid formula
  in $\CSL$, since it is the disjunction of two until operators.
\qed
\end{example}

In the remainder of this section, we shall focus on the proof of $\sim~=~\sim_{\CSLstar}$.
First, we introduce the following lemma in~\cite{SharmaK2012}:
\begin{lemma}[Theorem 5~\cite{SharmaK2012}]\label{lem:mtl cylinder}
  Given a path formula $\psi$ of \emph{$\CSLstar$} and a state $s$, there exists a set of cylinder
  sets $\mathit{Cyls}$ such that $\SAT(\psi)=\cup_{C\in\mathit{Cyls}}C.$
\end{lemma}
As a direct result of Lemma~\ref{lem:mtl cylinder}, $\SAT(\psi)$ is measurable
for any path formula $\psi$ of $\CSLstar$, as $\SAT(\psi)$ can be represented
by a countable set of measurable cylinders.

Now we are ready to present the main result of this section, i.e.,
strong bisimulation coincides with $\CSLstar$ equivalence for arbitrary $\CTMDP$s:
\begin{theorem}\label{thm:mtl strong}
  For any \emph{$\CTMDP$}, $\sim~=~\sim_{\CSLstar}$.
\end{theorem}

\subsection{Weak Bisimulation}\label{sec:mtl weak}
In this section we shall discuss the relation between weak bisimulation and
the equivalence induced by $\CSLstar_{\backslash\!\X}$.
Similar as in Section~\ref{sec:mtl strong} for strong bisimulation,
weak bisimulation can be fully characterized by $\CSLstar_{\backslash\!\X}$.

Since our weak bisimulation is defined as strong bisimulation
on the uniformized $\CTMDP$s, foremost we shall make sure that
$\CSLstar_{\backslash\!\X}$ is preserved by uniformization under TTP schedulers, that is,
we shall prove the following lemma:
\begin{lemma}\label{lem:uniformization mtl}
$s~\sim_{\CSLstar_{\backslash\!\X}}~r$ in $\MC{C}$ iff $\UNIFORM{s}~\sim_{\CSLstar}~\UNIFORM{r}$ in 
$\UNIFORM{\MC{C}}$.
\end{lemma}
As a side contribution, we extend the result in~\cite{NeuhausserSK09,RabeS11} and 
show that uniformization also does not change properties specified by $\CSLstar_{\backslash\!\X}$,
provided TTP schedulers are considered.
Given Lemma~\ref{lem:uniformization mtl}, the soundness and completeness of $\WBS$
with respect to $\sim_{\CSLstar_{\backslash\!\X}}$ are then straightforward
from Definition~\ref{def:weak bisimulation} and the fact that $\sim$ is both sound and 
complete with respect to $\CSLstar$.
\begin{theorem}\label{thm:mtl weak}
  For any \emph{$\CTMDP$}, $\WBS~=~\sim_{\CSLstar_{\backslash\!\X}}$.
\end{theorem}

Currently, we only prove Theorem~\ref{thm:mtl weak} with respect
to TTP schedulers. However, the optimal scheduler for a 
$\CSLstar$ formula may be not a TTP scheduler. Refer to the following example:
\begin{example}\label{ex:ttp insufficient}
Let $\MC{C}$ be a $\CTMDP$ as in Fig.~\ref{fig:cslstar}, where the letter on above of each state
denotes its label. Moreover states $s_8$ and $s_9$ only have self-loop transitions which are omitted.
Let $\psi=((a\lor b)\U^I d)\lor ((a\lor c)\U^I e)$ be a path formula of $\CSLstar$. We show
that there exists a non-TTP scheduler $\Sch$ such that
$$\PM_{\Sch,s_4}(\{\omega\in\Path{\INFPATH}(\MC{C})\mid\omega\models\psi\})>\PM_{\Sch',s_4}(\{\omega\in\Path{\INFPATH}(\MC{C})\mid\omega\models\psi\})$$
for any TTP scheduler $\Sch'$. Let $I=[0,\infty]$. Since $\Sch'$ is a TTP scheduler, it can only
make decision based on the elapsed time and the current state. When at $s_7$, 
$\Sch'$ will choose either the transition to $s_8$ or the transition to $s_9$ at each time point.
Therefore the maximal probability of satisfying $\psi$ is 0.5.
However for a general scheduler $\Sch$, it can make decision based on the full history.
For instance when at $s_7$, we can let $\Sch$ choose the transition to $s_8$, if the previous state is $s_5$,
otherwise $s_9$. Under this scheduler, the maximal probability of satisfying $\psi$ is equal to 1,
which cannot be obtained by any TTP scheduler. From this example, we can see that
an optimal scheduler for a $\CSLstar$ formula may make it decision based on the elapsed time
as well as the states visited. 
\end{example}

\begin{figure*}[!t]
\centering
  \begin{tikzpicture}[->,>=stealth,auto,node distance=2cm,semithick]
	\tikzstyle{state}=[circle,draw,thick]
	every label/.style=draw
        \tikzstyle{blackdot}=[circle,fill=black, minimum size=6pt,inner sep=0pt]
     \node[state,label={[label distance=0pt]90:{$a$}}](s4)[xshift=-1cm]{$s_4$};
     \node[blackdot](d)[right of=s4]{};
     \node[state,label={[label distance=0pt]90:{$b$}}](s5)[above right of=d,xshift=0.5cm]{$s_5$};
     \node[state,label={[label distance=0pt]90:{$c$}}](s6)[below right of=d,xshift=0.5cm]{$s_6$};
     \node[state,label={[label distance=0pt]90:{$a$}}](s7)[below right of=s5,xshift=0.5cm]{$s_7$};
     \node[state,label={[label distance=0pt]90:{$d$}}](s8)[above right of=s7,xshift=0.5cm]{$s_8$};
     \node[state,label={[label distance=0pt]90:{$e$}}](s9)[below right of=s7,xshift=0.5cm]{$s_9$};
     \path (s4) edge[-]             node {1} (d)
           (d)  edge[dashed]         node[left,yshift=5pt] {$\frac{1}{2}$} (s5)
                edge[dashed]         node[left,yshift=-5pt] {$\frac{1}{2}$} (s6)
           (s5) edge					node[right]{1} (s7)
           (s6) edge                  node[right]{1} (s7)
           (s7) edge                  node[left]{1}(s8)
           	edge                  node[left]{1}(s9);
  \end{tikzpicture}
  \caption{\label{fig:cslstar}
TTP schedulers are not enough to obtain optimal values for $\CSLstar$ properties.}
\end{figure*}
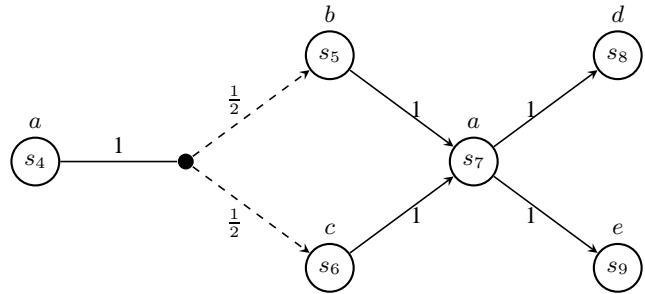

Example~\ref{ex:ttp insufficient} shows that it is not enough to
consider TTP schedulers in the setting of $\CSLstar$. 
In~\cite{NeuhausserSK09} another class of schedulers called \emph{Total
  Time History dependent schedulers} (TTH) is introduced.
We conjecture that for TTH schedulers: i) they preserve
$\CSLstar_{\backslash\!\X}$ properties after uniformization, and ii)
they are powerful enough to obtain optimal values for
$\CSLstar_{\backslash\!\X}$ properties. Condition i) guarantees
that Theorem~\ref{thm:mtl weak} is valid, while condition ii) makes 
Theorem~\ref{thm:mtl weak} general enough.
We leave the proof of the conjecture as our future work.

\begin{remark}
The expressiveness of $\CSLstar$ may be considered too powerful in certain cases.
For instance, path formulas like $\Box(a\U^{[2,10]} b)$
\footnote{$\Box\psi\equiv\neg((a\land\neg a)\U^{[0,\infty)}\neg\psi)$ for some $a$, i.e., $\psi$ holds forever.}
will be satisfied with probability 0 for any $\CTMDP$. In general, if $\psi$ can only be satisfied with
probability strictly less than 1, the probability of satisfying $\psi$ forever will be 0 for any $\CTMDP$.

In the other hand, a small fragment of $\CSLstar$ is enough to characterize strong bisimulation.
Let $\CSL^{\lor}$ denote the fragment of $\CSLstar$ whose path formulas are defined by the following syntax:
$\psi ::= \X^I\phi \mid \psi\lor\psi$. We have shown in~\cite{SongBTS2012}
that $\sim~=~\sim_{\CSL^{\lor}}$ for any $\CTMDP$. Therefore any subset of $\CSLstar$ which subsumes $\CSL^{\lor}$
will be strong enough to fully characterize strong bisimulation.
\end{remark}

\section{Relation to $\MDP$s and $\CTMC$s}\label{sec:relation}
In this section, we compare related work on other stochastic models: $\MDP$s and $\CTMC$s. 

\subsection{Relation to (Weak) Bisimulation for $\MDP$s}
For $\MDP$s, it is known that strong (probabilistic) bisimulation is only sound but not
complete with respect to $\PCTL$~\cite{SegalaL95}--the counterpart of $\CSL$ in discrete setting.
Differently, the completeness does not hold either even if we restrict to non 2-step recurrent 
$\MDP$s,   which can be defined in a straightforward way 
given Definition~\ref{def:2-step recurrent}. Refer to the following example:
\begin{example}
Let $s_0$ and $r_0$ be two states as in Example~\ref{ex:2-step recurrent ex1},
which will be viewed as two states in an $\MDP$. 
Moreover we assume that $u'_3$ only has a self loop. Since $u'_3$ has 
atomic propositions different from $s_0$ ($r_0$) and $u_3$, therefore $s_0$ and $r_0$
are not 2-step recurrent. However $s_0$ and $r_0$ satisfy the same $\PCTL$
formulas, since the maximal and minimal probabilities from $s_0$ and $r_0$ to any
subset of $\{u_1,u_2,u_3,u'_3\}$ are the same. As mentioned before, the middle transition
of $r_0$ cannot be simulated by any combined transition of $s_0$, hence
they are not strong probabilistic bisimilar. 
This indicates that strong (probabilistic) bisimulation is not
complete with respect to $\PCTL$ equivalence even that the given $\MDP$ is
non 2-step recurrent. \qed
\end{example}

The counterpart of $\CSLstar$ in discrete setting is
$\PCTL^*$~\cite{Aziz1995UWT}. Similar as in the continuous case, 
the equivalence induced by $\PCTL^*$
is strictly finer than $\sim_{\PCTL}$~\cite{SongZG2011}. 
However, different from the continuous case,
$\sim_{\PCTL^*}$ is still coarser than strong (probabilistic) bisimulation for $\MDP$s, that is,
strong (probabilistic) bisimulation is not complete with respect to $\PCTL^*$:
\begin{example}
Let $s_0$ and $r_0$ be two states as in 
Example~\ref{ex:counterexample completeness}, where we have shown that
$s_0$ and $r_0$ are neither strong bisimilar nor $\CSLstar$ equivalent.
However in~\cite{SongZG2011} $s_0$ and $r_0$ are shown to be
$\PCTL^*$ equivalent by viewing them as two states in an $\MDP$.
Therefore $\CSLstar$ gains more expressiveness by adding time bounds to the logic. 
\qed
\end{example}

The case for weak bisimulation is similar and omitted here.

\subsection{Relation to (Weak) Bisimulation for $\CTMC$s}\label{sec:relation to ctmc}

In this section we show that our bisimulations are downward compatible to those 
for $\CTMC$s. Different from $\CTMDP$s, there is no non-deterministic transitions in
$\CTMC$s, i.e., each state has only one transition, which will be
denoted by $s\TRAN{\lambda_s}\mu_s$. The notion of weak
bisimulation can be found in \cite{BaierKHW05} for $\CTMC$s, which is repeated
as follows:
\begin{definition}[Weak Bisimulation of $\CTMC$s]\label{def:weak bisimulation CTMC}
  For \emph{$\CTMC$}s, an equivalence relation $\MC{R}$ is a weak
  bisimulation iff for all $s~\MC{R}~r$ it holds: i) $L(s)=L(r)$, and
  ii) $\lambda_s\cdot\mu_s(C)=\lambda_r\cdot\mu_r(C)$ for all
  equivalence classes $C\neq \ESET{s}$.
 
 States $s,r$ are weak bisimilar, denoted by $s~\WBS_{\CTMC}~r$, iff
  there exists a weak bisimulation $\MC{R}$ such that $s~\MC{R}~r$.
\end{definition}

Strong bisimulation for $\CTMC$s is defined if in addition
$\lambda_s\cdot\mu_s(C)=\lambda_r\cdot\mu_r(C)$ holds for $C =
\ESET{s}= \ESET{r}$ as well.  States $s,r$ are strong bisimilar, denoted
by $s~\sim_{\CTMC}~r$, iff there exists a strong bisimulation $\MC{R}$
such that $s~\MC{R}~r$.

Below we prove that, restricted to $\CTMC$s, our strong and  weak bisimulations
agree with strong and weak bisimulations for $\CTMC$s, respectively:
\begin{lemma}\label{lem:weak bisimulation CTMC}
For $\CTMC$s, it holds  that $\sim~=~\sim_{\CTMC}$ and $\WBS~=~\WBS_{\CTMC}$.
\end{lemma}

The lemma above shows that $\sim$ and $\WBS$ are conservative extensions of strong and weak bisimulations for $\CTMC$s in \cite{BaierKHW05}, and so are their logical characterization results except that they only work on a subset of $\CTMDP$s free of 2-step recurrent states.

Since $\CTMC$s are sub-models of $\CTMDP$s,Theorem~\ref{thm:mtl strong} and \ref{thm:mtl weak} also hold for $\CTMC$s. 
Together with Lemma~\ref{lem:weak bisimulation CTMC}, we have the following result:
\begin{corollary}\label{cor:ctmc}
\begin{enumerate}
\item $\sim_{\CSLstar}~=~\sim~=~\sim_{\CTMC}~=~\sim_{\CSL}$,
\item $\sim_{\CSLstar_{\backslash\!\X}}~=~\WBS~=~\WBS_{\CTMC}~=~\sim_{\CSL_{\backslash\!\X}}$.
\end{enumerate}
\end{corollary}
Corollary~\ref{cor:ctmc} shows that $\CSLstar$ gains no more distinguishing power
than $\CSL$ on $\CTMC$s without non-determinism, similarly for their
sub-logics without the next operator.

\section{Conclusion and Future Work}\label{sec:conclusion}
In this paper, we have proposed both strong and weak bisimulations for
$\CTMDP$s, which are shown to be able to fully characterize $\CSL$ and
$\CSL_{\backslash\!\X}$ equivalences respectively, but over non 2-step
recurrent $\CTMDP$s. For a standard extension of $\CSL$ -- $\CSLstar$,
we show that strong and weak
bisimulations are both sound and complete with respect to
$\CSLstar$ and $\CSLstar_{\backslash\!\X}$
respectively for arbitrary $\CTMDP$s. Moreover, we give a simple scheme to determine non
2-step recurrent $\CTMDP$s, and show almost all $\CTMDP$s found in practice
are non 2-step recurrent $\CTMDP$s.  
We note that the work in this paper can
be extended to the simulation setting in a straightforward way.

For future work we would like to consider the approximation of
bisimulations and simulations on $\CTMDP$s as well as their logic
characterization, along \cite{DesharnaisGJP04}. Moreover,
the model checking of $\CSLstar$ against $\CTMC$s and $\CTMDP$s
will be also worthwhile to exploit. Another
interesting direction is to consider the continuous semantics of $\CSLstar$.

\section*{Acknowledgements}
Many thanks to the anonymous referees for their valuable suggestions
on an early version of this paper.  The authors are supported by
IDEA4CPS and the VKR Center of Excellence MT-LAB. The work has
received support from the EU FP7-ICT project MEALS (295261), and the
DFG Sonderforschungsbereich AVACS. Part of the work was done while the
first author was with IT University of Copenhagen, Denmark, and the
second author was with Technical University of Denmark.

\newpage
\appendix

\section{Proofs of Section~\ref{sec:bisimulation}}\label{app:proof}
\subsection{Proof of Theorem~\ref{thm:equivalent strong}}
\begin{proof}
Theorem~\ref{thm:martin}, states the direction $\sim~\subseteq~\sim_{\CSL}$, 
so here we only prove that $\sim_{\CSL}~\subseteq~\sim$. 

Let $\MC{R}=\sim_{\CSL}$ and $s~\MC{R}~r$, where
$\MC{R}$ is obviously an equivalence relation. We need to prove that
$L(s)=L(r)$ and for each $s\TRAN{\lambda}\mu$, there exists
$r\TRANP{\lambda}\mu'$ such that $\mu~\MC{R}~\mu'$.  The proof is
along the same line as the proof of Theorem~\ref{thm:mtl strong}: we only need to consider the $\ABS{\SUC(s)}>2$, as the
formula constructed there contains disjunctions.

Recall that in this theorem $\MC{C}$ is non  2-step recurrent.
Let $s_k\in\SUPP(\mu)$ be a state such that there exists $t\not\in\ESET[]{s}\cup\ESET[]{s_k}$, $s_k\TRAN{\lambda_2}\nu$ and $\nu(t)>0$ for some $\lambda_2$ and $\nu$. 
Since $\MC{C}$ is non  2-step recurrent,  such $s_k$ always exists.
Then the formula for case when $a_k\in(b_k,c_k)$ and
$a_j\in(c_j,b_j)$  is given by:
 \begin{align*}
\psi=(s \lor s_k)\U^{[a,b]}(s_j\lor t)
\end{align*}
We also distinguish the following three sub-cases:
 \begin{enumerate}
 \item[$\lambda_1=\lambda_2$:]
 Let
\begin{align*}
\rho_1&=\rho\cdot(e^{-\lambda_1a}-e^{-\lambda_1b}+a\lambda_1
e^{-\lambda_1a}-b\lambda_1e^{-\lambda_1b})\\
\rho_2&=(e^{-\lambda_1a}-e^{-\lambda_1b})
\end{align*}
then
\begin{itemize}
\item  the
probability of paths starting from $s$ satisfying $\psi$ by choosing transitions
$s\TRAN{\lambda_1}\mu$ and $s_k\TRAN{\lambda_2	`}\nu$ is equal to
$p(s,\mu):=a_j\cdot\rho_2+a_k\cdot\rho_1$, and
\item the probabilities of paths starting from $r$ satisfying $\psi$ by
choosing transitions  $r\TRAN{\lambda_1}\mu'_1$ and
$r\TRAN{\lambda_1}\mu'_2$ and then $s_k\TRAN{\lambda_2}\nu$ are equal to
$p(r,\nu_1):=b_j\cdot\rho_2+b_k\cdot\rho_1$ and $p(r,\nu_2)=c_j\cdot\rho_2+c_k\cdot\rho_1$ respectively.
\end{itemize}

As in Theorem~\ref{thm:mtl strong}, it is sufficient to prove that
$\frac{\rho_1}{\rho_2}\in(0,\infty)$, 
which can be seen as follows:
\begin{itemize}
\item Let $b=\infty$, then $\frac{\rho_1}{\rho_2}=\rho\cdot(a\lambda_1+1)$ and
it is easy to see that there exists $a,b$ such that
$\frac{\rho_1}{\rho_2}\in[\rho,\infty)$.
\item On the other hand let $a=0$,
then $\rho_1=\rho(1-e^{-\lambda_1b}-b\lambda_1e^{-\lambda_1b})$ and $\rho_2=1-e^{-\lambda_1b}$,
so $\frac{\rho_1}{\rho_2}=\rho\cdot(1-\frac{b\lambda_1
e^{-\lambda_1b}}{1-e^{-\lambda_1b}})$, note here that $\frac{b\lambda_1
e^{-\lambda_1b}}{1-e^{-\lambda_1b}}\in(0,1)$ since $\frac{b\lambda_1
e^{-\lambda_1b}}{1-e^{-\lambda_1b}}$ can be arbitrary close to 1 when $b$ is close to 0, 
while $\frac{b\lambda_1 e^{-\lambda_1b}}{1-e^{-\lambda_1b}}$ is arbitrary close to 0 as $b$ increases. 
As a result $\frac{\rho_1}{\rho_2}\in(0,\rho)$.
\end{itemize}
\item[$\lambda_1>\lambda_2$:]Then $\rho_2$ will be the same as in the case when $\lambda_1=\lambda_2$ 
and 
$$\rho_1=\rho(\frac{\lambda_1}{\lambda_1-\lambda_2}
(e^{-\lambda_2a}-e^{\lambda_2b}) - \frac{\lambda_2}{\lambda_1-\lambda_2}(e^{-\lambda_1a}-e^{\lambda_1b})).$$
Therefore
$$\frac{\rho_1}{\rho_2}=\rho(\frac{\lambda_1}{\lambda_1-\lambda_2}(\frac{e^{-\lambda_2a}-e^{-\lambda_2b}}{e^{-\lambda_1a}-e^{-\lambda_1b}})-\frac{\lambda_2}{\lambda_1-\lambda_2}).$$
When $\lambda_1>\lambda_2$, $\frac{e^{-\lambda_2a}-e^{-\lambda_2b}}{e^{-\lambda_1a}-e^{-\lambda_1b}}\in(\frac{\lambda_2}{\lambda_1},\infty)$, thus $\frac{\rho_1}{\rho_2}\in(0,\infty)$.
The remaining arguments are the same as in the case when $\lambda_1=\lambda_2$.
\item[$\lambda_1<\lambda_2$:] This case is similar as the above case and is omitted.

%Then $\rho_2$ will be the same as in the case when $\lambda_1=\lambda_2$ 
%and 
%$$\rho_1=\rho(\frac{\lambda_2}{\lambda_2-\lambda_1}
%(e^{-\lambda_1a}-e^{\lambda_1b}) - \frac{\lambda_1}{\lambda_2-\lambda_1}(e^{-\lambda_2a}-e^{\lambda_2b})).$$
%As in the above clause, we compute the range of 
%$$\frac{\rho_1}{\rho_2}=\rho(\frac{\lambda_2}{\lambda_2-\lambda_1}-\frac{\lambda_1}{\lambda_2-\lambda_1}(\frac{e^{-\lambda_2a}-e^{-\lambda_2b}}{e^{-\lambda_1a}-e^{-\lambda_1b}})).$$
%Note that $\frac{e^{-\lambda_2a}-e^{-\lambda_2b}}{e^{-\lambda_1a}-e^{-\lambda_1b}}\in(0,\frac{\lambda_2}{\lambda_1})$, thus $\frac{\rho_1}{\rho_2}\in(0,\frac{\rho\lambda_2}{\lambda_2-\lambda_1})$. 
\end{enumerate}
Thus there always exists $0\le a\le b$ such that paths starting from $r$ will satisfy
$\psi$ with higher probability than $s$ for some $a,b$, therefore
$s~\not\sim_{\CSL}~r$, which contradict the assumption.
\qed
\end{proof}

\subsection{Proof of Lemma~\ref{lem:sim_properties}}
\begin{proof}
\begin{enumerate}
\item Let $\MC{R}=\sim$ and $s~\MC{R}~r$. 
To show that $\sim$ implies $\WBS$, it is enough to prove that $\MC{R}$
is a weak bisimulation.
Let $\UNIFORM{s}\TRAN{E}\mu$, according to Definition~\ref{def:weak bisimulation}
we need to prove that there exists $\UNIFORM{r}\TRAN{E}\nu$ such that $\mu~\MC{R}~\nu$.
By Definition~\ref{def:uniformization}, $\UNIFORM{s}\TRAN{E}\mu$ iff there exists $s\TRAN{\lambda}\mu'$
such that $\mu=\frac{\lambda}{E}\cdot\mu' + \frac{E-\lambda}{E}\cdot\DIRAC{s}$. Since $s~\sim~r$, there
exists $r\TRANP{\lambda}\nu'$ such that $\mu'~\MC{R}~\nu'$. Note that $r\TRANP{\lambda}\nu'$ implies
$r\TRANP{E}\nu\equiv\frac{\lambda}{E}\cdot\nu' + \frac{E-\lambda}{E}\cdot\DIRAC{r}$, apparently $\mu~\MC{R}~\nu$
as required.
\item The proof of Clause 2 is straightforward from Definition~\ref{def:weak bisimulation}.
\end{enumerate}
\qed
\end{proof}

\subsection{Proof of Lemma~\ref{lem:CSL preserved}}
\begin{proof}
We first prove that if $\MC{C}$ is a $\CTMC$, then 
$s~\sim_{\CSL_{\backslash\!\X}}~r$ in $\MC{C}$ 
iff $\UNIFORM{s}~\sim_{\CSL}~\UNIFORM{r}$ in $\UNIFORM{\MC{C}}$.
Since uniformization preserves the satisfiability of $\CSL_{\backslash\!\X}$, we have $\UNIFORM{s}~\sim_{\CSL_{\backslash\!\X}}~\UNIFORM{r}$.

Let $\MC{R}=\sim_{\CSL_{\backslash\!\X}}$ and $\UNIFORM{s}~\MC{R}~\UNIFORM{r}$.
According to~\cite{BaierKHW05}, $\CSL$ equivalence coincides with strong bisimulation 
on $\CTMC$s, therefore it suffices to prove that $\MC{R}$ is a strong bisimulation.
Let $\lambda$ denote the exit rate of $\UNIFORM{s}$ and $\UNIFORM{r}$,
and $\lambda_{\UNIFORM{s}}$ denote the rate from $\UNIFORM{s}$ to
states in $\ESET{\UNIFORM{s}}$ i.e. $\lambda_{\UNIFORM{s}}=\lambda\cdot\mu(\ESET{\UNIFORM{s}})$
where $\UNIFORM{s}\TRAN{\lambda}\mu$. We need to prove that there exists $\UNIFORM{r}\TRAN{\lambda}\nu$
such that $\mu~R~\nu$.

The case when $\lambda_{\UNIFORM{s}}=\lambda$ is trivial,
we assume that $\lambda>\lambda_{\UNIFORM{s}}$.

In the following proof, we let $\phi_C$ be a formula such that $\SAT(\phi_C)=C$ where $C$ is
a $\MC{R}$ closed set.
Now we are going to prove that $\lambda_{\UNIFORM{s}}=\lambda_{\UNIFORM{r}}$ i.e. the rates for $s$ and $r$
leaving to states in equivalence classes different from $\ESET{\UNIFORM{s}}$
are equal. Let $C=\UNIFORM{\MC{S}}\setminus\ESET{\UNIFORM{s}}$,
then $s\models\MC{P}_{\geq p}(\phi_{\ESET{\UNIFORM{s}}}\U^{[a,b]}\phi_C)$
where $p=e^{-\lambda'a}-e^{-\lambda'b}$ and $\lambda'=\lambda-\lambda_{\UNIFORM{s}}$.
Since $\UNIFORM{s}~\sim_{\CSL_{\backslash\!\X}}~\UNIFORM{r}$, we have
$r\models\MC{P}_{\geq p}(\phi_{\ESET{\UNIFORM{s}}}\U^{[a,b]}\phi_C)$ for any $0\leq a<b$.
Therefore $\lambda-\lambda_{\UNIFORM{s}}=\lambda-\lambda_{\UNIFORM{r}}$ which implies
$\lambda_{\UNIFORM{s}}=\lambda_{\UNIFORM{r}}$.

Let $C\in \UNIFORM{S}/\MC{R}$ be an equivalence relation such that $\UNIFORM{s}\notin C$,
we know that $\UNIFORM{s}\models\phi:=\MC{P}_{\geq p}(\phi_{\ESET{\UNIFORM{s}}}\U^{[a,b]}\phi_C)$
where
$$p=\frac{\lambda\cdot\mu(C)}{\lambda-\lambda_{\UNIFORM{s}}}\cdot(e^{-\lambda_C\cdot a} - e^{-\lambda_C\cdot b}).$$
Since $\UNIFORM{s}~\sim_{\CSL_{\backslash\!\X}}~\UNIFORM{r}$,
we have $\UNIFORM{r}\models\phi$. We show that it must be the case that
$\mu(C)=\nu(C)$. We prove by contradiction and distinguish
the following cases:
\begin{enumerate}
\item $\mu(C)<\nu(C)$. Let $a=0$ and $b=\infty$, then
$p=\frac{\lambda\cdot\mu(C)}{\lambda-\lambda_{\UNIFORM{s}}}.$
The probability of the paths starting from $r$ satisfying $(\phi_{\ESET{\UNIFORM{s}}}\U^{[a,b]}\phi_C)$ is
$\frac{\lambda\cdot\nu(C)}{\lambda-\lambda_{\UNIFORM{r}}}$ which
is apparently greater than $p$, given that we have proved that $\lambda_{\UNIFORM{s}}=\lambda_{\UNIFORM{r}}$.
Therefore $r\models\MC{P}_{\geq p'}(\phi_{\ESET{\UNIFORM{s}}}\U^{[a,b]}\phi_C)$,
but $s\not\models\MC{P}_{\geq p'}(\phi_{\ESET{\UNIFORM{s}}}\U^{[a,b]}\phi_C)$ where $p'=\frac{\lambda\cdot\nu(C)}{\lambda-\lambda_{\UNIFORM{r}}}$,
 this contradicts with our assumption.
\item $\mu(C)>\nu(C)$. This case is similar as the first case by letting $a=0$ and $b=\infty$, thus is omitted here.
\end{enumerate}
Consequently, we have that $\mu(C)=\nu(C)$ for each $C\in \UNIFORM{S}/\MC{R}$ except for $\ESET{\UNIFORM{s}}$,
moreover $\lambda_{\UNIFORM{s}}=\lambda_{\UNIFORM{r}}$, hence $\mu~\MC{R}~\nu$ and $\MC{R}$ is a strong bisimulation.
According to \cite{BaierKHW05} where it is was shown that $\sim$ is both sound and complete for $\sim_{\CSL}$ on $\CTMC$,
thus $\UNIFORM{s}~\sim_{\CSL}~\UNIFORM{r}$.

We now generalize the result to $\CTMDP$s. If $s~\sim_{\CSL_{\backslash\!\X}}~r$, then $\UNIFORM{s}~\sim_{\CSL_{\backslash\!\X}}~\UNIFORM{r}$. Since in a uniformized $\CTMDP$, every execution of $\MC{C}$ guided by a given scheduler can be seen as a $\CTMC$, thus $\UNIFORM{s}~\sim_{\CSL}~\UNIFORM{r}$ based on the above result.
\qed
\end{proof}

\subsection{Proof of Theorem~\ref{thm:equivalent weak}}
\begin{proof}
Since in Theorem~\ref{thm:equivalent strong}, we have shown that
$\sim~=~\sim_{\CSL}$ provided that $\MC{C}$ is non 2-step recurrent.
The proof is straightforward since:

$
(s~\WBS~r)~ \stackrel{Def.~\ref{def:weak bisimulation}}{\Longleftrightarrow}~
(\UNIFORM{s}~\sim~\UNIFORM{r})~ \stackrel{Thm.~\ref{thm:equivalent strong}}{\Longleftrightarrow}~
(\UNIFORM{s}~\sim_{\CSL}~\UNIFORM{r})~ \stackrel{Lem.~\ref{lem:CSL preserved}}{\Longleftrightarrow}~
(s~\sim_{\CSL_{\backslash\!\X}}~r).
$
\qed
\end{proof}

\subsection{Proof of Lemma~\ref{lem:2-step recurrent}}
\begin{proof}
Straightforward from Definition~\ref{def:2-step recurrent}.
The first two cases are simple since they do not depend on the given relation. 
We only need to check the third condition.
Since $\MC{R}~\subseteq~\MC{R}'$ implies $\ESET{s}~\subseteq~\ESET[\MC{R}']{s}$
for any $s$. Therefore if there exists
$s\TRAN{\lambda}\mu$ such that for all $s'\in\SUPP(\mu)$ and
$s'\TRAN{\lambda'}\nu$, we always have $\nu(C)=1$ where $C=\ESET{s}\cup\ESET{s'}$,
it must be the case that $\nu(C')=1$ where $C'=\ESET[\MC{R}']{s}\cup\ESET[\MC{R}']{s'}$,
since $C~\subseteq~C'$.
\qed
\end{proof}

\subsection{Proof of Lemma~\ref{lem:NP result}}
\begin{proof}
Our proof is inspired by  the reduction used in the long version of \cite{Desharnais2011computing}.  We sketch the proof here.

Consider the subset sum problem which is known to be NP-hard~\cite{Cormen2001IA}:
Given a set of $n$ integers $\{k_1,\ldots,k_n\}$, is there a non-empty subset whose sum is equal to 0. 
Note any subset sum problem can be reduced to the following problem by dividing each $k_i$ by 
$\frac{1}{4n}\cdot\max\{\ABS{k_i}\}$ where $1\le i\le n$:
Given $n$ decimal numbers $w_1,\ldots,w_n$ such that
$w_i\in[-\frac{1}{4n},\frac{1}{4n}]$ for each $i\in[1,n]$, can we find a set $I\subseteq[1,n]$ such that
$\sum_{i\in I}w_i=0$. We show that this problem can also be transformed to a problem of deciding 
the negation of $\sim_{\CSL}$ by constructing a $\CTMDP$ as follows:
Suppose we have states $s_0$, $s'_0$, $r$, and $\{s_i\}_{1\leq i\leq n}$, all of which have distinct
atomic propositions except $L(s_0)=L(s'_0)$, and moreover they only have a self loop transition with 
rate 1 except: 
$s_0\TRAN{1}\mu$, $s'_0\TRAN{1}\nu_1$, and $s'_0\TRAN{1}\nu_2$, where for each $1\leq i\leq n$
\begin{itemize}
\item $\mu(s_i)=|w_i| + \epsilon$ with $\epsilon=10^{-2n}$,
\item $\nu_1(s_i) = w_i+|w_i|$,
\item $\nu_2(s_i) = -w_i + |w_i|$.
\end{itemize}
Moreover let $\mu(r)=1-\sum_{1\leq i\leq n}(|w_i|+\epsilon)$,
$\nu_1(r)=1-\sum_{1\leq i\leq n}(w_i+|w_i|)$, and $\nu_2(r)=1-\sum_{1\leq i\leq n}(-w_i + |w_i|)$.
Clearly $\mu$, $\nu_1$, and $\nu_2$ are full distributions. 
In order to check whether $s_0~\sim_{\CSL}~s'_0$, the only non-trivial cases are formulas like
$\MC{P}_{\geq p}(\psi)$, where $\psi=\top\U^{[a,b]}(\lor_{s\in C}s)$ for some 
$C\subseteq\{s_i\}_{1\leq i\leq n}\cup\{r\}$. Since the probabilities of paths starting from $s_0$ and $s'_0$
satisfying $\psi$
by choosing transitions to $\mu$, $\nu_1$, and $\nu_2$ are equal to:
$\mu(E)\cdot(e^{-a}-e^{-b})$, $\nu_1(E)\cdot(e^{-a}-e^{-b})$, and $\nu_2(E)\cdot(e^{-a}-e^{-b})$
respectively, $s_0~\not\sim_{\CSL}~s'_0$ iff they exists $E$ such that $\mu(E)>\nu_1(E)$ and 
$\mu(E) > \nu_2(E)$. We distinguish the following two cases:
\begin{enumerate}
\item $r\not\in E$ i.e. there exists $I\subseteq[1,n]$ such that $E=\{s_i\mid i\in I\}$.\\
In this case we will have 
$$\sum_{i\in I}\mu(s_i) > \sum_{i\in I}\nu_1(s_i),\ \sum_{i\in I}\mu(s_i) > \sum_{i\in I}\nu_2(s_i),$$ 
which implies 
$$\sum_{i\in I}(\epsilon+|w_i|)>\sum_{i\in I}(w_i+|w_i|),\
\sum_{i\in I}(\epsilon+|w_i|)>\sum_{i\in I}(-w_i+|w_i|),$$ 
which implies
$$\sum_{i\in I}w_i < \epsilon\cdot|E|,\ 
-\sum_{i\in I}w_i < \epsilon\cdot |E|.$$
Since $\epsilon\cdot |E| < 10^{-2n}\cdot n < \frac{1}{4n}$, the only possibility 
for both $\sum_{i\in I}w_i < \epsilon\cdot|E|$ and $-\sum_{i\in I}w_i < \epsilon\cdot |E|$ hold
is that $\sum_{i\in I}w_i = 0$. 
\item $r\in E$ i.e. there exists $I\subseteq[1,n]$ such that $E=\{s_i\mid i\in I\}\cup\{r\}$.\\
In this case we will have 
$$\mu(r) + \sum_{i\in I}\mu(s_i) >\nu_1(r) + \sum_{i\in I}\nu_1(s_i),$$ 
$$\mu(r) + \sum_{i\in I}\mu(s_i) > \nu_2(r) + \sum_{i\in I}\nu_2(s_i),$$
which implies
$$1-\sum_{1\leq i\leq n}(\epsilon+|w_i|) + \sum_{i\in I}(\epsilon+|w_i|)>1-\sum_{1\leq i\leq n}(w_i+|w_i|)+\sum_{i\in I}(w_i+|w_i|),$$
$$1-\sum_{1\leq i\leq n}(\epsilon+|w_i|) + \sum_{i\in I}(\epsilon+|w_i|)>1-\sum_{1\leq i\leq n}(-w_i + |w_i|) + \sum_{i\in I}(-w_i+|w_i|),$$
which implies
$$-\epsilon\cdot|\bar{I}| > \sum_{i\in \bar{I}} w_i,$$
$$-\epsilon\cdot|\bar{I}| > -\sum_{i\in\bar{I}}w_i,$$
where $\bar{I}=[1,n]\setminus I$, which holds iff
$\bar{I}=\emptyset$, but this contradicts that $\mu(E)=\nu_1(E)=\nu_2(E)=1$.
\end{enumerate}
In conclusion, $s_0~\not\sim_{\CSL}~s'_0$ iff there exist $I\subseteq[1,n]$ such that
$\sum_{i\in I}w_i=0$. Since the reduction is polynomial, we can say that it is NP-hard to decide 
$\not\sim_{\CSL}$, which implies that the decision of $\sim_{\CSL}$
is also NP-hard.  

The above proof can also be applied to prove that 
deciding $\sim{\CSL_{\backslash\!\X}}$ is NP-hard.
\qed
\end{proof}

\section{Proofs of Section~\ref{sec:mtl}}

\subsection{Proof of Theorem~\ref{thm:mtl strong}}
The proof of Theorem~\ref{thm:mtl strong} is divided into the following lemmas:
\begin{lemma}\label{lem:mtl soundness}
  $s~\sim~r$ implies $s~\sim_{\CSLstar}~r$ for any $s$ and $r$ i.e. $\sim~\subseteq~\sim_{\CSLstar}$. 
\end{lemma}
\begin{proof}
    We shall show that $s~\sim~r$ implies $s~\sim_{\CSLstar}~r$ for any $s$ and $r$, that is,
    $s~\sim~r$ and $s\models\phi$ implies that $r\models\phi$ for any $\phi$.     
    Given two cylinders $C_1$ and $C_2$, 
    we say that $C_1$ and $C_2$ are strong bisimilar, written as
    $C_1~\sim~C_2$, iff $\Len{C_1}=\Len{C_2}$, $\iPath{C_1}{i}~\sim~\iPath{C_2}{i}$ 
    for each $0\le i\le\Len{C_1}$, and $\Time(C_1,i)=\Time(C_2,i)$ for each $0\le i<\Len{C_1}$.
    Similarly, we can define strong bisimulation of paths.

    As usual we prove the following two things simultaneously:
    \begin{enumerate}
	\item \label{itm:first}$s\models\phi$ iff $r\models\phi$ for any $\phi$, provided that $s~\sim~r$;
	\item \label{itm:second}$\omega_1\models\psi$ iff $\omega_2\models\psi$ for any $\psi$, provided that
	$\omega_1~\sim~\omega_2$.
    \end{enumerate}
    We only show the proof for case when $\phi=\MC{P}_{\ge q}(\psi)$ and 
    $\psi=\psi_1\U^I\psi_2$, since all the other cases are either trivial or similar.
    Suppose that $s\models\phi$ i.e. for all schedulers $\Sch$,
    $\PM_{\Sch,s}(\SAT(\psi))\ge q$, we shall prove that $\PM_{\Sch,r}(\SAT(\psi))\ge q$ for
    any scheduler $\Sch$ of $r$.
    According to Lemma~\ref{lem:mtl cylinder}, the set of paths starting from $s$
    and satisfying $\psi$ can be represented by a set of cylinders $\mathit{Cyls}$. 
    By induction hypothesis, $\SAT(\psi)$ is $\sim$ closed, thus for any $C\in\mathit{Cyls}$,
    $\ESET[\sim]{C}\subseteq\SAT(\psi)$. 
    Since for any $C_1,C_2\in\mathit{Cyls}$ such that $C_1\cap C_2\neq\emptyset$,
    there exists a set of disjoint cylinders $\{C'_i\}$ such that $\cup\{C'_i\}=C_1\cup C_2$,
    so any $\mathit{Cyls}$ can be transformed to an equivalent set of disjoint cylinders.
    In the sequel we assume that $\mathit{Cyls}$ contains only disjoint cylinders,
    therefore 
    $$\PM_{\Sch,s}(\{\omega\in\Path{\infty}\mid\omega\models\psi\})=\sum_{C\in\mathit{Cyls}}\PM_{\Sch,s}(C),$$
    for any scheduler $\Sch$.
    As a result, it suffices to prove that for each scheduler $\Sch_1$ of $s$, 
    there exists a scheduler $\Sch_2$ of $r$ such that
    $\PM_{\Sch_1,s}(\ESET[\sim]{C})=\PM_{\Sch_2,r}(\ESET[\sim]{C})$ for each $C\in\mathit{Cyls}$.
    Let $C=C(s_0,I_0,\ldots,I_{n-1},s_n)$ where $s_0=s$, we prove by induction on $n$. 
    The base case when $n=0$ is trivial. Assume that $n>0$, then
    according to Eq.~(\ref{eq:measure}), $\PM_{\Sch_1,s}(\ESET[\sim]{C})=\PM_{\Sch_1,s}(\ESET[\sim]{C},0)=$
    $$
    \int\limits_{t\in I_0}\sum_{(\lambda,\mu)\in\MI{tr}}\Sch(s,0)(\lambda,\mu)\cdot\sum_{s'\in\ESET[\sim]{s_1}}\mu(s')
    \cdot\lambda e^{-\lambda t}\cdot\PM_{\Sch,s'}(\ESET[\sim]{C'},t)dt,
    $$
    where $\MI{tr}=\NEXT(s)$ and $C'=C(s_1,I_1,\ldots,s_n)$.
    Since $s~\sim~r$, for each $(\lambda,\mu)\in\MI{tr}$ there exists $r\TRANP{\lambda}\nu$
    such that $\mu~\sim~\nu$. Let $\Sch_2$ mimic exactly what $\Sch_1$ does when at state $r$. 
    Moreover $\PM_{\Sch_1,s'}(\ESET[\sim]{C'},t)=\PM_{\Sch_2,r'}(\ESET[\sim]{C'},t)$
    for each $C'\in\mathit{Cyls}$ such that $\Len{C'}<n$,
    provided $s'~\sim~r'$. By induction hypothesis, such $\Sch_2$ always exists, and
    $\PM_{\Sch_1,s}(\ESET[\sim]{C})=\PM_{\Sch_2,r}(\ESET[\sim]{C})$ for each $C$. Consequently, we have
    $r\models\phi$. \qed
\end{proof}

\begin{lemma}\label{lem:mtl completeness}
$s~\sim_{\CSLstar}~r$ implies $s~\sim~r$ for any $s$ and $r$ i.e. $\sim_{\CSLstar}~\subseteq~\sim$.
\end{lemma}
\begin{proof}
First we define a sub-logic of $\CSLstar$, called $\CSL^{\lor}$, whose
state formulas are the same as $\CSLstar$, while its path formulas are
defined by the following BNFs:
$$
\psi ::= \X^I\phi \mid \psi\lor\psi,
$$
that is, the only path formula of $\CSL^{\lor}$ is the disjunction of several 
next operators. 

Secondly, we prove that $\sim_{\CSL^{\lor}}~\subseteq~\sim$. 
Let $\MC{R}=\{(s,r)\mid s~\sim_{\CSL^{\lor}}~r\}$ and $s~\MC{R}~r$, where
$\MC{R}$ is obviously an equivalence relation. The proof of
$L(s)=L(r)$ is trivial and omitted here. It suffices now to prove that
for each $s\TRAN{\lambda}\mu$, there exists
$r\TRANP{\lambda}\mu'$ such that $\mu~\MC{R}~\mu'$. 

\emph{Claim.}
Fix a $s\TRAN{\lambda}\mu$,  there exists $r\TRAN{\lambda}\mu'$
such that $\mu(C)=\mu'(C)=1$ for some $\MC{R}$-closed set $C$.

To prove the claim we let
$\{\lambda_i\mid r\TRAN{\lambda_i}\mu'_i\land\mu'_i(C)=1\}_{1\leq i\leq n}$. 
We proceed by contradiction and assume that there does not exist $i$ such that $\lambda_i=\lambda$. Without
loss of generality, we assume that $n=2$. There are three cases we should consider here:
\begin{enumerate}
\item $\lambda_1 < \lambda_2 < \lambda$. Let $\phi_C$ be a formula such that $\SAT(\phi_C)=C$, 
since $C$ is $\MC{R}$ closed, $\phi_C$ always exists. Let $\psi = \X^{[0,b]}\phi_C$, then the
maximal probability of paths starting from $s$ satisfying $\psi$ is equal to $1-e^{-\lambda_2\cdot b}$,
while the probability for $r$ is $1-e^{-\lambda\cdot b}$ which is obviously less than 
$1-e^{-\lambda_2\cdot b}$. Therefore there exists $p=1-e^{-\lambda_2\cdot b}$, such that
$s\models\MC{P}_{\leq p}(\psi)$, but $r\not\models\MC{P}_{\leq p}(\psi)$, which contradicts
the assumption that $s~\sim_{\CSL^{\lor}}~r$. 
\item $\lambda < \lambda_1 < \lambda_2$. This case is similar with the above case and omitted here. 
\item $\lambda_1 < \lambda <\lambda_2$. 
%Similarly, we can show that the probabilities 
%of paths starting from $s$ and $r$ satisfying $\psi=\X^{[a,b]}\phi_C$ by choosing transitions with rates 
%$\lambda_1$, $\lambda_2$, and $\lambda$ are equal to $e^{-\lambda_1\cdot a} - e^{-\lambda_1\cdot b}$,
%$e^{-\lambda_2\cdot a} - e^{-\lambda_2\cdot b}$, and $e^{-\lambda\cdot a} - e^{-\lambda\cdot b}$ 
%respectively. 
Let $f(x)=e^{-ax}-e^{-bx}$, then $df/dx = b\cdot e^{-bx} - a\cdot e^{-ax}$. 
We solve the inequation $df/dx > 0$, and get $x < \ln(b/a)/(b-a)$, which means that
if $x_1 < x_2 \leq \ln(b/a)/(b-a)$ or $x_1 > x_2 \geq \ln(b/a)/(b-a)$, we have 
$$e^{-ax_2}-e^{-bx_2} > e^{-ax_1}-e^{-bx_1}.$$
Let $a,b$ be two real numbers such that $\lambda=\frac{ln(b/a)}{b-a}$, thus it holds that $$e^{-\lambda\cdot a} - e^{-\lambda\cdot b} > \max\{e^{-\lambda_1\cdot a} - e^{-\lambda_1\cdot b}, e^{-\lambda_2\cdot a} - e^{-\lambda_2\cdot b}\}.$$
Therefore there also exists $p$ such that $s\models\MC{P}_{\leq p}(\psi)$, 
but $r\not\models\MC{P}_{\leq p}(\psi)$, which contradicts the assumption. Thus, we have the claim.
\end{enumerate}

To proceed with the proof of the main theorem, we show that for each $s\TRAN{\lambda_1}\mu$,
there exists $r\TRANP{\lambda_1}\mu'$ such that $\mu~\MC{R}~\mu'$. 
Due to the above proven claim, it is enough to focus on transitions with 
same rates. 
We proceed by contradiction, and assume there exists a set of transitions
$\{\mu'_i\mid r\TRAN{\lambda_1}\mu'_i\}$ with
$1\leq i\leq n$, but there does not exist $\{w_i\in [0,1]\}$ 
such that $\mu~\MC{R}~\mu'$ where $\mu'=\sum_{1\leq i\leq n} w_i\cdot\mu'_i$. 
In order to get a contradiction, we need to find a formula $\phi$ which is satisfied 
by $s$ but not $r$, or the other way around. We consider the following cases:
\begin{enumerate}
\item $\ABSORB{s}$ i.e. $s$ is a silent state. This case is impossible 
since all the derivations of $s$ will stay in the same equivalence class 
$\ESET[]{s}$, as well as $r$, thus there exists $r\TRANP{\lambda_1}\nu$ such that $\mu(\ESET[]{s})=\nu(\ESET[]{s})=1$.
\item
$\SUC(s)\leq 2$ i.e. there exists at most two equivalence classes
$C_1,C_2\subseteq C$ such that $\mu(C_1\cup C_2) = 1$, in other
words, $\mu(C_1)=1-\mu(C_2)$. In case of $\SUC(s)$ is a singleton set,
we simply set $C_2=\emptyset$. We consider the following cases:
\begin{enumerate}
\item $\mu'_1(C_1) \leq \mu'_2(C_1) < \mu(C_1)$. Let $\psi=\X^{[0,\infty)}\phi_{C_1}$,
the maximal probability of paths starting from $r$ satisfying $\psi$ is $\mu(C_1)$,
while the maximal probability for $s$ is $\mu'_2(C_1)$ less than $\mu(C_1)$,
thus there exists $p$ such that $s\models\MC{P}_{\leq p}(\psi)$, but 
$r\not\models\MC{P}_{\leq p}(\psi)$, which contradict the assumption.
\item $\mu'_2(C_1) \geq \mu'_1(C_1) > \mu(C_1)$. This case is similar with
the case above, and is omitted here.
\item $\mu'_1(C_1) \leq \mu(C_1) \leq \mu'_2(C_1)$. In this case we can make sure that
there exists $w_1,w_2$ such that $w_1+w_2=1$ and $w_1\cdot\mu'_1(C_1) +
w_2\cdot\mu'_2(C_1)=\mu(C_1)$, therefore
\begin{align*}
w_1\cdot\mu'_1(C_2) +
w_2\cdot\mu'_2(C_2) 
&=w_1\cdot(1-\mu'_1(C_1)) + w_2\cdot(1- \mu'_2(C_1))\\
&= w_1+w_2-(w_1\cdot\mu'_1(C_1) + w_2\cdot\mu'_2(C_1))\\
&=1-\mu'(C_1)=\mu'(C_2)
\end{align*} 
thus $(w_1\cdot\mu'_1+w_2\cdot\mu'_2)=\mu'$ such that $\mu~\MC{R}~\mu'$ as we expect. 
Note this cannot be generalized to the case when $\SUC(s)>2$. 
\end{enumerate}
\item
We consider the -- most involved -- remaining case:
$\SUC(s)>2$.
Note that every combined transition of $r$ can be seen as a
combined transition of two other (combined) transitions of $r$.
We fix two arbitrary (combined) transitions of $r$: 
$r\TRANP{\lambda_1}\mu'_1$ and $r\TRANP{\lambda_1}\mu'_2$, thus
\begin{eqnarray}\label{eq:all}
\forall 0\leq w_1,w_2\leq 1. & &w_1 + w_2 = 1 \\\notag
&\wedge&\mu~\not\MC{R}~(w_1\cdot\mu'_1 +
    w_2\cdot\mu'_2).
\end{eqnarray}
Let $\SUPP(\mu)=\{s_1,s_2,\ldots,s_n\}$. For simplicity we assume
that $s_1,\ldots,s_n$ belong to different equivalence classes.
For $1\leq i\leq n$, define: $\mu(s_i)=a_i, \mu'_1(s_i)=b_i, \mbox{
and }\mu'_2(s_i)=c_i .$
According to Eq.~(\ref{eq:all}), for each $k$ there must exist $1\leq j\neq k\leq n$ 
such that there does not exist $0\leq w_1,w_2\leq 1$ with
$w_1+w_2=1$ such that $w_1\cdot b_k + w_2\cdot c_k=a_k$ and
$w_1\cdot b_j + w_2\cdot c_j=a_j$, otherwise $\mu~\MC{R}~(w_1\mu'_1+w_2\mu'_2)$ 
which contradicts Eq.~(\ref{eq:all}). The idea now is then to construct a formula $\phi$ which is
satisfied by  $s$ but not $r$.
There are several cases to be considered 
depending on whether $a_k\in [b_k,c_k]$ and/or $a_j\in[b_j,c_j]$.
Most of the cases are trivial except when $a_k\in(b_k,c_k)$ and
$a_j\in(c_j,b_j)$ with $c_k\geq b_k$ and $b_j\geq
c_j$. For instance if $a_k>b_k,c_k$,
$s$ will evolve into $s_k$ with higher probability than $r$, so $\phi$
is easy to give. 

Let $\psi := (\XI[{[a,b]}]s_j)\lor (\XI[{[a',b']}]s_k),$ 
 where the names of states are used as abbreviations of the state formulas characterizing the  equivalence classes where they are located.
Then the probability of paths starting from $s$ satisfying $\psi$ by 
choosing transition
$s\TRAN{\lambda_1}\mu$ is equal to
$p(s,\mu):=a_j\cdot\rho_2+a_k\cdot\rho_1$,
where $\rho_1=(e^{-\lambda_1 a} - e^{-\lambda_1 b})$ and
$\rho_2=(e^{-\lambda_1 a'} - e^{-\lambda_1 b'})$. 
Similarly, the probabilities of paths starting from $r$ satisfying $\psi$ by
choosing transitions  $r\TRAN{\lambda_1}\mu'_1$ and
$r\TRAN{\lambda_1}\mu'_2$ are equal to
$p(r,\mu'_1):=b_j\cdot\rho_2+b_k\cdot\rho_1$ and $p(r,\mu'_2)=c_j\cdot\rho_2+c_k\cdot\rho_1$ respectively.

Now it is sufficient to prove that we can always find $0\le a\le b$ and 
$0\le a'\le b'$ such that
$p(s,\mu)>\max\{p(r,\mu'_1),p(r,\mu'_2)\}$. 
\begin{enumerate}
\item $\frac{b_j-a_j}{a_k-b_k} < \frac{a_j-c_j}{c_k-a_k}$:
Let $\frac{\rho_1}{\rho_2}\in(\frac{b_j-a_j}{a_k-b_k},\frac{a_j-c_j}{c_k-a_k})$,
then we have $a_k\cdot\rho_1 + a_j\cdot\rho_2 >\max\{b_k\cdot\rho_1 + b_j\cdot\rho_2, c_k\cdot\rho_1 + c_j\cdot\rho_2\}$
i.e. $p(s,\mu)>\max\{p(r,\mu'_1),p(r,\mu'_2)\}$ as we shall prove. 
Note that $\frac{\rho_1}{\rho_2}=\frac{e^{-\lambda_1 a} - e^{-\lambda_1 b}}{e^{-\lambda_1 a'} - e^{-\lambda_1 b'}}$ ranges over $[0,\infty)$ by choosing
different values for $a,b,a'$, and $b'$, therefore the discriminating formula
always exists, we get contradiction.   The case when $\frac{b_j-a_j}{a_k-b_k}>\frac{a_j-c_j}{c_k-a_k}$ can be proved in a similar way, and is omitted here.
\item $\frac{b_j-a_j}{a_k-b_k}=\frac{a_j-c_j}{c_k-a_k}$:
This case is impossible, otherwise there
exists $0\leq w_1,w_2\leq 1$ such that $w_1\cdot b_k + w_2\cdot
c_k=a_k$ and $w_1\cdot b_j + w_2\cdot c_j=a_j$ with $w_1+w_2=1$,
simply let $w_1=\frac{1}{k+1}$ and $w_2=\frac{k}{k+1}$ where $k=\frac{a_k-b_k}{c_k-a_k}$. 
\end{enumerate}
\end{enumerate}

Since $\CSL^{\lor}$ is a sub-logic of $\CSLstar$, trivially $\sim_{\CSLstar}~\subseteq~\sim_{\CSL^{\lor}}$,
therefore $\sim_{\CSLstar}~\subseteq~\sim$, which completes the proof.
\qed
\end{proof}

\subsection{Proof of Lemma~\ref{lem:uniformization mtl}}
\begin{proof}
Since in Lemma~\ref{lem:mtl cylinder} we have shown that for any $s$ and $\psi$,
the paths starting from $s$ and satisfying $\psi$ can be represented by
a set of disjoint cylinders. It suffices to prove that for each $\Sch$ of $s$,
$\PM_{\Sch,s}(C)=\PM_{\UNIFORM{\Sch},\UNIFORM{s}}(\UNIFORM{C})$
for each cylinder $C$, where $\UNIFORM{C}$ is a cylinder same as $C$ 
except that $\iPath{\UNIFORM{C}}{i}=\UNIFORM{\iPath{C}{i}}$ for each $0\le i\le\Len{C}$,
and $\UNIFORM{\Sch}$ is the scheduler mimicking $\Sch$ stepwise. 
Let $C=s_0,I_0,s_1,\ldots,s_n$, we shall prove by induction on $n$.
The case when $n=0$ is trivial, since $\PM_{\Sch,s}(C)$ is either 1 or 0 depending on
whether $s_0=s$. Suppose that $n>0$, $s_0=s$, and $I_0=[a,b]$,
Since it has been proved in~\cite[Sec. 6]{RabeS11} that uniformization does not 
change time-bounded reachability, that is, 
the probability from $s_0$ to $s_1$ in time interval $I$
is equal to the probability from $\UNIFORM{s_0}$ to $\UNIFORM{s_1}$ in time interval $I$
for any $I$. Let $F(t)$ denote the probability from $s_0$ to $s_1$ in time interval $[0,t]$ given scheduler $\Sch$,
and $f(t)=\frac{dF(t)}{dt}$, that is, $f(t)$ is the corresponding probability density function, 
similarly we can define $\UNIFORM{F}(t)$ and $\UNIFORM{f}(t)$.
According to Eq.~(\ref{eq:measure}), $\PM_{\Sch,s}(C)=\PM_{\Sch,s}(C,0)=
\int_{t\in I_0}f(t)\cdot\PM_{\Sch,s}(C,t)dt$
and $\PM_{\UNIFORM{\Sch},\UNIFORM{s}}(C)=\PM_{\UNIFORM{\Sch},\UNIFORM{s}}(C,0)=
\int_{t\in I_0}\UNIFORM{f}(t)\cdot\PM_{\UNIFORM{\Sch},\UNIFORM{s}}(C,t)dt.
$
Since $F(t)=\UNIFORM{F}(t)$ for any $t$, we have $f(t)=\UNIFORM{f}(t)$ for any $t$. 
By induction hypothesis, $\PM_{\Sch,s}(C,t)=\PM_{\UNIFORM{\Sch},\UNIFORM{s}}(C,t)$
for any $t$, thus 
$$f(t)\cdot\PM_{\Sch,s}(C,t)=\UNIFORM{f}(t)\cdot\PM_{\UNIFORM{\Sch},\UNIFORM{s}}(C,t)$$
for any $t$, which indicates that $\PM_{\Sch,s}(C)=\PM_{\UNIFORM{\Sch},\UNIFORM{s}}(C)$.
\qed
\end{proof}

\subsection{Proof of Section~\ref{thm:mtl weak}}
\begin{proof}
  The proof can be presented as the following chain:
  $$
  s~\WBS~r \stackrel{Def.~\ref{def:weak bisimulation}}{\Longleftrightarrow}
  \UNIFORM{s}~\sim~\UNIFORM{r} \stackrel{Thm.~\ref{thm:mtl strong}}{\Longleftrightarrow}
  \UNIFORM{s}~\sim_{\CSLstar}~\UNIFORM{r} \stackrel{Lem.~\ref{lem:uniformization mtl}}{\Longleftrightarrow}
  s~\sim_{\CSLstar_{\backslash\!\X}}~r.
  $$
\end{proof}

\section{Proofs of Section~\ref{sec:relation to ctmc}}
\subsection{Proof of Lemma~\ref{lem:weak bisimulation CTMC}}
\begin{proof}
The proof of $\sim~=~\sim_{\CTMC}$ is trivial, since in a $\CTMC$ there is only one transition for each state, thus we can simply replace $\TRANP{}$ with $\TRAN{}$. The condition $\lambda_s\cdot\mu_s(C)=\lambda_r\cdot\mu_r(C)$ for each $C$ coincides with the condition: i) $\lambda_s=\lambda_r$, and ii) $\mu_s~\MC{R}~\mu_r$.

We first prove that $\WBS$ implies $\WBS_{\CTMC}$. Let $\MC{R}=\WBS$ and $s~\MC{R}~r$. We shall prove
that $\MC{R}$ is a weak bisimulation as defined in Definition~\ref{def:weak bisimulation CTMC}. Suppose that $s\TRAN{\lambda_s}\mu_s$, we need to prove that $r\TRAN{\lambda_r}\mu_r$ such that $\lambda_s\cdot\mu_s(C)=\lambda_r\cdot\mu_r(C)$ for all $C\in S/\MC{R}$ with $C\neq\ESET{s}=\ESET{r}$. According to Definition~\ref{def:weak bisimulation}, $s~\WBS~r$ if $\UNIFORM{s}~\iBS{}~\UNIFORM{r}$. By Definition~\ref{def:uniformization}, if $s\TRAN{\lambda_s}\mu_s$, then $\UNIFORM{s}\TRAN{E}\mu$ such that $\mu=\frac{E-\lambda_s}{E}\cdot\DIRAC{\UNIFORM{s}} + \frac{\lambda_s}{E}\cdot\UNIFORM{\mu_s}$ where $\UNIFORM{\mu_s}$ is defined as expected. Therefore there exists $\UNIFORM{r}\TRAN{E}\nu$ such that $\mu~\iBS{}~\nu$ where $\nu=\frac{E-\lambda_r}{E}\cdot\DIRAC{\UNIFORM{r}} + \frac{\lambda_r}{E}\cdot\UNIFORM{\mu_r}$. Obviously if there exists $C\in S/\MC{R}$ with $C\neq\ESET{s}=\ESET{r}$ such that $\lambda_s\cdot\mu_s(C)\neq\lambda_r\cdot\mu_r(C)$, then $\mu(\UNIFORM{C})\neq\nu(\UNIFORM{C})$ since $\mu(\UNIFORM{C})=\frac{\lambda_s}{E}\cdot\mu_s(C)$ and $\nu(\UNIFORM{C})=\frac{\lambda_r}{E}\cdot\mu_r(C)$, thus it is impossible for $\mu~\iBS{}~\nu$.

To show that $\WBS_{\CTMC}$ implies $\WBS$, it is enough to show that $\MC{R}=\WBS_{\CTMC}$ is a weak bisimulation according to Definition~\ref{def:weak bisimulation}, that is, we need show that $\MC{R}=\{(\UNIFORM{s},\UNIFORM{r})\mid s~\WBS_{\CTMC}~r\}$ is a strong bisimulation by Definition~\ref{def:strong bisimulation}. Suppose that $\UNIFORM{s}\TRAN{E}\mu$, then there exists $s\TRAN{\lambda_s}\mu_s$ such that $\mu=\frac{E-\lambda_s}{E}\cdot\DIRAC{\UNIFORM{s}} + \frac{\lambda_s}{E}\cdot\UNIFORM{\mu_s}$. Since $s~\WBS_{\CTMC}~r$, there exists $r\TRAN{\lambda_r}\mu_r$ such that $\lambda_s\cdot\mu_s(C)=\lambda_r\cdot\mu_r(C)$ for all equivalence class $C\neq \ESET[\WBS_{\CTMC}]{s}=\ESET[\WBS_{\CTMC}]{r}$. Therefore there exists $\UNIFORM{r}\TRAN{E}\nu$ such that $\nu=\frac{E-\lambda_r}{E}\cdot\DIRAC{\UNIFORM{r}} + \frac{\lambda_r}{E}\cdot\UNIFORM{\mu_r}$ and $\mu(\UNIFORM{C})=\nu(\UNIFORM{C})$ for all equivalence class $\UNIFORM{C}\neq \ESET{\UNIFORM{s}}=\ESET{\UNIFORM{r}}$, since $\mu(\UNIFORM{C})=\frac{\lambda_s}{E}\cdot\mu_s(C)$ and $\nu(\UNIFORM{C})=\frac{\lambda_r}{E}\cdot\mu_r(C)$ i.e. $\mu~\MC{R}~\nu$.
\qed
\end{proof}

\end{document}